\documentclass[superscriptaddress, amsmath,amssymb,
aps, pra, letterpaper, tightenlines, reprint, notitlepages]{revtex4-2}

\usepackage{graphicx}
\usepackage{floatrow}

\usepackage{dcolumn}% Align table columns on decimal point
\usepackage{bm}% bold math

\usepackage{amsfonts}
\usepackage{dsfont}
\usepackage{bm}
\usepackage{bbold}
\usepackage{xcolor}
\usepackage{lipsum,babel}
\usepackage[normalem]{ulem}

\usepackage[utf8]{inputenc}
\usepackage{amsmath}
\usepackage{physics}
\usepackage{hyperref}
\usepackage{amssymb}
\usepackage{braket}
\usepackage{mathtools}
\usepackage{color}
\usepackage{amsthm}
\usepackage{algorithm}
\usepackage{algorithmic}

\usepackage{lipsum}

\DeclareMathOperator*{\E}{\mathbb{E}}

\DeclareMathOperator*{\poly}{\rm{poly}}

\theoremstyle{plain}
\newtheorem{theorem}{Theorem}
\newtheorem{corollary}{Corollary}

\newtheorem{definition}{Definition}
\newtheorem{lemma}{Lemma}
\newtheorem{conjecture}{Conjecture}
\newtheorem{remark}{Remark}

\theoremstyle{definition}

\begin{document}

\title{Hardness and Complexity Transition of Noisy Random Circuit Sampling}
%Classical Hardness of Random Circuit Sampling under Local Depolarizing Noise

\author{Byeongseon Go}
\email{gbs1997@snu.ac.kr}
\affiliation{NextQuantum Center, Department of Physics and Astronomy, Seoul National University, Seoul 08826, Republic of Korea}
\author{Changhun Oh}
\email{changhun0218@gmail.com}
\affiliation{Department of Physics, Korea Advanced Institute of Science and Technology, Daejeon 34141, Republic of Korea}
\author{Hyunseok Jeong}
\email{h.jeong37@gmail.com}
\affiliation{NextQuantum Center, Department of Physics and Astronomy, Seoul National University, Seoul 08826, Republic of Korea}

\begin{abstract}
Random circuit sampling (RCS) is a leading candidate for demonstrating quantum advantage, supported by strong complexity-theoretic evidence of hardness in the ideal setting and by rapid experimental progress to date. 
In practice, however, noise is unavoidable, and a central problem is to identify the noise-strength boundary between classically simulable and classically hard regimes.
In this work, we establish an architecture-general hardness bound for this boundary for the standard local depolarizing noise of strength $\gamma$.
Assuming the standard average-case \#P-hardness conjecture for ideal RCS, we show that, for any circuit architecture satisfying this conjecture, noisy RCS on the same architecture remains hard to simulate classically within any inverse-polynomial total variation distance whenever $\gamma=O(\log n/(nd))$ for $n$-qubit circuits of depth $d$, unless the polynomial hierarchy collapses.
Crucially, noisy-RCS hardness follows without any additional conjectural or architecture-specific assumption beyond those already entering the ideal-RCS hardness framework.
Our proof combines a low-degree polynomial extrapolation with a monotonicity reduction showing that efficient classical simulation at one depolarizing noise strength implies efficient simulation at every larger strength.
Together, these ingredients transfer the standard ideal-RCS hardness conjecture to sampling hardness at a prespecified noise strength.
Finally, combining the convergence-to-uniformity result of Dalzell \textit{et al.} [Commun.~Math.~Phys.~405, 78 (2024)] with our monotonicity reduction yields efficient classical simulation for $\gamma=\omega(\log n/(nd))$ on layered, regularly connected architectures.
Thus, wherever the two architectural settings overlap, this identifies $\gamma=\Theta(\log n/(nd))$ as the asymptotic complexity-transition scale.

\end{abstract}

\maketitle

\section{Introduction}

Demonstrating quantum advantage is a central milestone in the development of quantum computation.
Among the leading proposals, random circuit sampling~(RCS)~\cite{bouland2019complexity, bouland2022noise, kondo2022quantum, movassagh2023hardness, krovi2022average} has emerged as a particularly prominent candidate,  supported by strong complexity-theoretic evidence of classical hardness in the ideal setting and by substantial experimental progress.
These features have motivated large-scale experimental implementations of RCS~\cite{arute2019quantum, wu2021strong, zhu2022quantum,  morvan2024phase, decross2025computational, gao2025establishing}.

%In realistic implementations, however, noise is unavoidable, and its effect on the computational complexity of RCS remains partially understood.
In realistic implementations, however, noise is unavoidable, making it crucial to understand how the computational complexity of RCS depends on the noise strength.
A large body of work has identified noise regimes in which RCS can be efficiently simulated classically~\cite{aharonov1996limitations, gao2018efficient, deshpande2022tight, aharonov2023polynomial, dalzell2024random, nelson2025limitations, nelson2026polynomial, lee2025classical, zhang2025classically, noh2020efficient, zhang2023noisy, chen2018classical, cirstoiu2024fourier, huang2020classical, hangleiter2023computational, schuster2025polynomial, pan2022simulation, fontana2025classical, muller2024enabling, cheng2021simulating, ayral2023density, lee2025scalable}.
By contrast, much less is known about the complementary question of how weak the noise must be for noisy RCS to retain the hardness of its ideal counterpart.
Determining this boundary is essential both for understanding the complexity transition induced by noise and for establishing rigorous hardness guarantees for noisy RCS.

In this work, we establish an architecture-general hardness result for noisy RCS under local depolarizing noise.
Let $\mathcal{A}$ be any circuit architecture for which the standard average-case \#P-hardness conjecture for ideal RCS holds.
For $n$-qubit circuits of depth $d$, we show that no polynomial-time classical sampler can simulate noisy RCS over $\mathcal{A}$ within any prescribed inverse-polynomial total variation distance~(TVD) at a prespecified noise strength $\gamma^*$ satisfying
\begin{align}
\gamma^*
=
O\left(\frac{\log n}{nd}\right),
\end{align}
unless the polynomial hierarchy collapses.
Crucially, the reduction from ideal to noisy RCS introduces no additional hardness assumption.
Once a circuit architecture satisfies the stated ideal-RCS hardness conjecture, noisy RCS on the same architecture inherits sampling hardness up to the corresponding noise scale.
Accordingly, the theorem is not tied to a specially constructed architecture and does not require a separate hardness conjecture for the noisy output distribution.
%Thus, our result is not tied to a specially constructed architecture: every architecture supporting the standard ideal-RCS hardness conjecture automatically inherits noisy sampling hardness up to the stated noise scale.

The significance of this noise-strength scale is clearest when compared with Ref.~\cite{dalzell2024random}.
For architectures that both satisfy the ideal-RCS hardness conjecture and fall within the layered, regularly connected setting covered by that work, their convergence-to-uniformity result, together with our monotonicity theorem, gives efficient classical simulation for
\begin{align}
\gamma
=
\omega\left(\frac{\log n}{nd}\right).
\end{align}
\iffalse
For architectures that both satisfy the ideal-RCS hardness conjecture and fall within the layered, regularly connected setting covered by that work, their convergence-to-uniformity result gives efficient classical simulation for
\begin{align}
    \gamma=\omega\left(\frac{\log n}{nd}\right).
\end{align}
\fi
Thus, within the intersection of the two architectural regimes, the known hardness and simulability bounds meet at the same asymptotic noise scale.

The proof combines a low-degree extrapolation that recovers ideal output probabilities from noisy evaluations with a monotonicity reduction showing that efficient simulation at one depolarizing noise strength implies efficient simulation at every larger strength.
The extrapolation produces the scale $nd\gamma^*=O(\log n)$, while the monotonicity reduction converts the resulting variable-noise probability-estimation hardness into sampling hardness at the prespecified noise strength $\gamma^*$.

Our result provides a distinct route to noisy-RCS hardness.
We directly transfer the standard ideal-RCS hardness conjecture to the exact local-depolarizing output distribution, yielding approximate-sampling hardness within inverse-polynomial TVD for any architecture satisfying the ideal-RCS hardness assumption.
This differs from previous approaches based on direct hardness of noisy output-probability estimation or on a reduction through a white-noise description.
A detailed comparison is given in Sec.~\ref{section: phase transition}.

\begin{figure*}[t]
\includegraphics[width=0.92\linewidth]{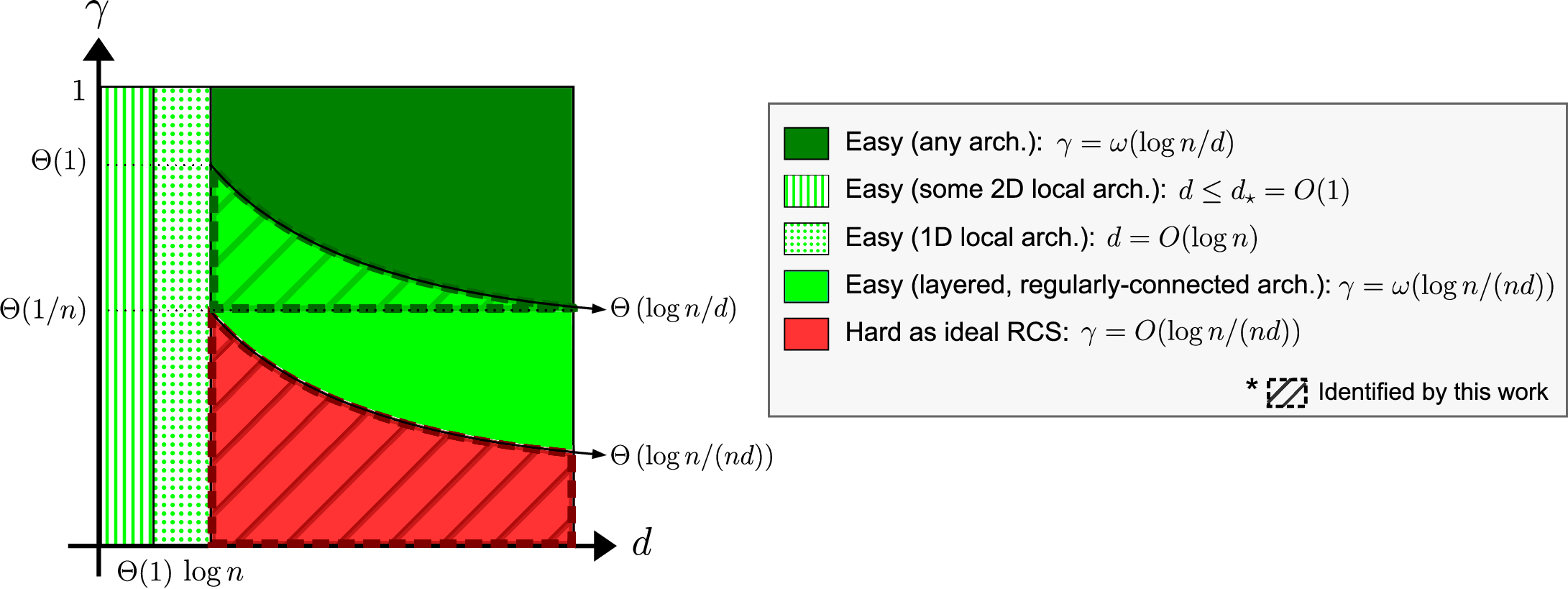}
\caption{Complexity phase diagram for simulating (sampling) noisy RCS under i.i.d. depolarizing noise (Definition~\ref{def: noisy quantum circuit}) within inverse-polynomial TVD.
The horizontal axis denotes the circuit depth $d$, and the vertical axis denotes the depolarizing noise strength $\gamma$ in Definition~\ref{def: noisy quantum circuit}.
Colored regions indicate parameter regimes where, for certain circuit architectures, classical simulation is either known to be efficient or provably hard under the hardness conjecture for ideal RCS (Conjecture~\ref{conj: average-case hardness}).
The dark green region is implied by Refs.~\cite{aharonov1996limitations, gao2018efficient, muller2016relative, mele2024noise, stilck2021limitations}.
The hatched and dotted regions are obtained from known shallow-depth simulability results for RCS: Refs.~\cite{napp2022efficient, cheng2023efficient, chen2024optimized} for some local 2D architectures, and Refs.~\cite{vidal2004efficient, markov2008simulating, jozsa2006simulation} for local 1D architectures.
The light green region follows from Ref.~\cite{dalzell2024random}: the unhatched part is covered directly by their result, and the hatched part follows by applying the monotonicity result in Theorem~\ref{thm: complexity decrease}.
Finally, the red region follows from Theorem~\ref{thm: hardness informal}.
The different regions may rely on different architectural assumptions
and should not be interpreted as simultaneous statements for every single architecture.
The red and light-green bounds form a matching asymptotic boundary only for architectures satisfying both sets of hypotheses.
}
\label{fig: transition}
\end{figure*}

\iffalse

Moreover, we show that the complexity of noisy RCS is \textit{monotonically nonincreasing} with the depolarizing noise strength $\gamma$.
That is, efficient classical simulation at a smaller noise strength implies efficient classical simulation at any larger noise strength.
%, in the sense that the simulability of noisy RCS for a smaller noise strength implies the simulability of noisy RCS for any larger noise strength.
This monotonicity is not only a key ingredient in our hardness proofs, but also propagates existing simulability results.
In particular, whenever ideal RCS, or noisy RCS at a smaller noise strength, is efficiently simulable for a given architecture (e.g., shallow-depth RCS~\cite{napp2022efficient}), noisy RCS at any larger depolarizing noise strength is also efficiently simulable.
We expect that this monotonicity principle will further contribute to a deeper understanding of the complexity of noisy RCS.

\fi

The remainder of this paper is organized as follows.
In Sec.~\ref{section: phase transition}, we summarize our main hardness results for noisy RCS, review related work, and clarify our contributions.
In Sec.~\ref{section: preliminaries}, we collect the definitions and notations used throughout this work.
In Sec.~\ref{section: main problem and result}, we formulate the main computational problems building on prior hardness arguments for ideal RCS
%we formally define our main computational problem building on prior arguments for ideal RCS, 
and introduce our main hardness result for noisy RCS.
In Sec.~\ref{section: proof of theorem 1}, we prove Theorem~\ref{thm 1}, establishing the hardness result for noisy RCS.
%In Sec.~\ref{section: constructing architectures}, we construct circuit architectures primarily used to prove Theorem~\ref{thm 2}, which establishes our second hardness result for noisy RCS.
%In Sec.~\ref{section: proof of theorem 2}, we construct circuit architectures primarily used to prove Theorem~\ref{thm 2}, and then prove Theorem~\ref{thm 2}.
In Sec.~\ref{section: proof of theorem 3}, we prove Theorem~\ref{thm: complexity decrease}, which is crucial for translating our hardness results into the final simulation-hardness statement summarized in Theorem~\ref{corol: hardness}.
%In Sec.~\ref{section: numerical evidence}, we present numerical evidence for anti-concentration of noisy RCS, a key ingredient in our hardness arguments.
Finally, in Sec.~\ref{section: conclusion}, we conclude with remarks and directions for future work.

\section{Main result and relation to prior work}
\label{section: phase transition}

In this section, we summarize our main hardness result and compare it with the most closely related hardness and simulability results for noisy RCS.
Our result has two principal implications.
First, the noisy-RCS hardness theorem relies on exactly the same average-case hardness conjecture as ideal RCS.
The reduction requires no additional complexity-theoretic assumption, anti-concentration assumption for the noisy output distribution, or architecture-specific scrambling condition.
Accordingly, any architecture satisfying the stated ideal-RCS hardness conjecture inherits approximate-sampling hardness under local depolarizing noise up to the threshold established here.
Second, when combined with the existing simulability result, our theorem identifies an asymptotically matching complexity-transition scale within the overlapping architectural regime.

%First, it transfers the conjectured hardness of ideal RCS directly to the hardness of noisy RCS, without relying on architecture-specific scrambling assumptions.
%Second, within the class of architectures covered by the existing simulability result, our hardness bound identifies an asymptotically tight threshold for the noise strength.

%\subsection{Architecture-general hardness and a tight noise threshold}

\subsection{Architecture-general hardness and a matching transition scale}

We first state our main result on the hardness of classically simulating noisy RCS under the local depolarizing noise model of
Definition~\ref{def: noisy quantum circuit}.

\begin{theorem}[Informal]\label{thm: hardness informal}
Suppose there exists a circuit architecture for which ideal RCS satisfies the standard average-case \#P-hardness conjecture~\cite{bouland2019complexity, bouland2022noise, kondo2022quantum, movassagh2023hardness, krovi2022average} of output probability estimation, stated later as Conjecture~\ref{conj: average-case hardness}.
Then, noisy RCS on this architecture is classically hard to simulate within inverse-polynomial TVD whenever the depolarizing noise strength satisfies
\begin{align}\label{eq: noise threshold informal}
\gamma^* = O\left(\frac{\log n}{nd}\right),
\end{align}
for $n$-qubit circuits of depth $d$, unless the polynomial hierarchy collapses.
\end{theorem}

The significance of the theorem lies in the fact that its noisy-RCS hardness conclusion requires no conjectural assumption beyond Conjecture~\ref{conj: average-case hardness} for ideal RCS.
Hence, Theorem~\ref{thm: hardness informal} is architecture-general in the following sense: within the local Haar-random circuit setting of Definition~\ref{def: random circuit distribution}, it imposes no additional structural assumptions on the circuit architecture beyond those entering the ideal-RCS hardness conjecture.
%such as a layered or regularly connected architecture, global scrambling, or anti-concentration of the noisy output distribution.
%In particular, it does not separately require 

The threshold in Theorem~\ref{thm: hardness informal} becomes particularly informative when combined with the simulability result of Ref.~\cite{dalzell2024random}.
For layered, regularly connected architectures satisfying the assumptions of that work, their convergence-to-uniformity theorem establishes efficient average-case classical simulation in the regime
\begin{align}
\gamma
=
\omega\left(\frac{\log n}{nd}\right)
\end{align}
within its stated weak-noise range $\gamma = O(1/n)$.
Our monotonicity theorem then extends this simulability conclusion to every larger depolarizing-noise strength.
Consequently, for architectures lying in the intersection of the two settings, the known hardness and simulability bounds meet at the same asymptotic noise scale.

\begin{corollary}
\label{cor: transition scale}
Consider a circuit architecture that both satisfies
Conjecture~\ref{conj: average-case hardness} and lies within the regime
covered by Ref.~\cite{dalzell2024random}.
Then noisy RCS is classically hard for $\gamma = O(\log n / (nd))$, whereas it is classically simulable for $\gamma = \omega(\log n / (nd))$.
Hence, within this common regime, $\gamma=\Theta(\log n/(nd))$ is the asymptotic complexity-transition scale at the level of its dependence on $n$ and $d$.
\end{corollary}

Thus, within the overlapping architectural regime, the dependence of our hardness scale on $n$ and $d$ is asymptotically optimal.
We emphasize that monotonicity extends only the simulability conclusion: it does not ensure that the original convergence bound of Ref.~\cite{dalzell2024random} remains valid outside its proven weak-noise regime.

The origin of this scaling can be seen directly from our reduction.
For an $n$-qubit circuit of depth $d$, the noisy output probability is a polynomial of degree at most $N=n(d+1)$ in the depolarizing-noise strength $\gamma$, as shown in Eq.~\eqref{eq: noisy output probability}.
We construct a low-degree approximation to this polynomial and extrapolate from evaluations at nonzero noise strengths $\gamma\in[\gamma^*,1]$ to $\gamma=0$, where the noisy probability coincides with the ideal probability.
The extrapolation incurs an overhead of $e^{O(nd\gamma^*)}$, which remains polynomial when $nd\gamma^*=O(\log n)$.
This gives \#P-hardness of noisy-probability estimation when $\gamma\in[\gamma^*,1]$ is supplied as an input. 
Stockmeyer's reduction~\cite{stockmeyer1985approximation} then gives hardness of the corresponding variable-noise sampling problem, and our monotonicity theorem converts this into sampling hardness at the prespecified noise strength $\gamma^*$.

\iffalse
The same monotonicity result also strengthens the simulability statement in Corollary~\ref{cor: transition scale}.
Although the convergence bound of Ref.~\cite{dalzell2024random} is established only in the noise regime $\gamma=O(1/n)$, monotonicity implies efficient simulation at every larger depolarizing-noise strength once simulability is established at one such value.
It therefore extends the simulable region beyond $\gamma=O(1/n)$, as illustrated in Fig.~\ref{fig: transition}, without implying that the original convergence bound remains valid outside its proven parameter regime.
\fi

Finally, more general convergence-to-uniformity results apply without the architectural conditions of Ref.~\cite{dalzell2024random}, but yield the weaker simulability threshold~\cite{aharonov1996limitations, gao2018efficient, muller2016relative, mele2024noise, stilck2021limitations}
\begin{align}
\gamma
=
\omega\left(\frac{\log n}{d}\right)  .
\end{align}
The general simulable regime described above, together with the matching boundary obtained by combining Theorem~\ref{thm: hardness informal} and the simulability result of Ref.~\cite{dalzell2024random}, is summarized in Fig.~\ref{fig: transition}.

%\subsection{Strengthening the white-noise-based hardness approach}
\subsection{Direct hardness transfer beyond the white-noise route}

%Ref.~\cite{dalzell2024random} also provides the most directly related route from the conjectured hardness of ideal RCS to hardness of noisy RCS under local noise as ours.
Ref.~\cite{dalzell2024random} provides the most closely related previous route from ideal-RCS hardness to hardness under physical local noise.
Their approach first approximates the output distribution under local noise by a global white-noise distribution of the form $p_{\rm wn}=Fp_{\rm ideal}+(1-F)p_{\rm unif}$, where $F$ denotes the remaining ideal signal and $p_{\mathrm{unif}}$ is the uniform distribution.
Their appendix then relates the approximate sampling from
$p_{\rm wn}$ to the approximate sampling from the ideal distribution.

When this approach is used to establish hardness for noisy RCS under local noise, the resulting hardness conclusion inherits all conditions required for the local-to-white-noise approximation. 
These include the relevant structural and scrambling conditions on the circuit architecture, an appropriate anti-concentration condition, and a sufficiently weak-noise regime $\gamma \ll 1/(n\log n)$.
%in which the approximation error remains small relative to the remaining ideal-signal weight $F$.

Our result bypasses this intermediate white-noise approximation and the restrictions associated with that step.
Specifically, the reduction acts directly on the output probabilities of noisy RCS and therefore does not require the distribution to be approximated by a global white-noise mixture.
Consequently, the resulting hardness theorem does not inherit the architectural assumptions or the weak-noise condition required by the white-noise-based approach.
Instead, it applies to any circuit architecture satisfying the stated average-case hardness conjecture for ideal output-probability estimation.
This distinction is particularly useful in the logarithmic depth regime $d=\Theta(\log n)$, where our result establishes hardness for $\gamma=O(1/n)$, whereas the hardness implication of Ref.~\cite{dalzell2024random} requires $\gamma\ll 1/(n\log n)$.

The two approaches begin from differently formulated hardness assumptions for ideal RCS: Ref.~\cite{dalzell2024random} is based on hardness of approximate ideal-RCS sampling, whereas our theorem assumes average-case \#P-hardness of ideal output-probability estimation. 
We therefore do not claim a formal implication between the two final hardness statements.
The strengthening established here concerns the transfer from ideal to locally noisy RCS: our reduction bypasses the intermediate white-noise description and the associated architectural and weak-noise restrictions.

\iffalse

Our result is therefore not merely an alternative derivation of the same noise threshold.
Instead, it recovers the asymptotically tight hardness scaling in the regime already accessible through the white-noise approach, while extending the hardness statement beyond the architecture- and noise-dependent conditions under which the local-to-white-noise approximation has been established.
This distinction is particularly useful in the logarithmic depth regime $d=\Theta(\log n)$, where our result establishes hardness for $\gamma=O(1/n)$, whereas the  implication of Ref.~\cite{dalzell2024random} requires $\gamma\ll 1/(n\log n)$.
At the level of structural assumptions, our theorem thus strengthens and generalizes the previous white-noise-based route to noisy-RCS hardness.

We note that the two approaches formulate their underlying ideal-RCS hardness assumptions differently.
Ref.~\cite{dalzell2024random} assumes the hardness of approximate sampling from the ideal RCS distribution, whereas our result is based on the standard average-case \#P-hardness conjecture for estimating ideal output probabilities.
We therefore do not claim a general logical implication between these two complexity assumptions.
Rather, the improvement established here lies in removing the additional architectural, scrambling, and white-noise-approximation conditions required to transfer ideal-circuit hardness to the physical local-noise distribution.
\fi

%\subsection{A complementary hardness regime}
\subsection{Comparison with direct noisy-probability hardness}

Ref.~\cite{bouland2022noise} establishes hardness of estimating noisy output probabilities under a constant, gate-independent stochastic noise below an error-detection threshold.
The result applies to circuit architectures that support the required error-detection construction and proves hardness at an exponentially fine additive-error scale of $\exp(-O(m\log m))$, where $m$ denotes the number of gates in the circuit.
While the result tolerates constant noise strength for hardness, its approximation scale is much finer than the
$2^{-n}/\poly(n)$ scale required in the standard Stockmeyer reduction for approximate-sampling hardness, and it therefore does not by itself establish hardness of sampling noisy RCS within inverse-polynomial TVD.

Our result addresses the complementary direction.
We restrict to local depolarizing noise satisfying $\gamma = O\left({\log n}/({nd})\right)$, but establish a complete approximate-sampling hardness statement within any prescribed inverse-polynomial TVD.
%but complete approximate-sampling hardness statement within any prescribed inverse-polynomial TVD.
Moreover, our theorem applies to any circuit architecture satisfying the stated hardness conjecture for ideal-RCS output-probability estimation and does not require a separate error-detection construction.
The two results therefore address complementary regimes in terms of noise strength, architectural assumptions, and approximation accuracy.

\section{Preliminaries}\label{section: preliminaries}

\subsection{Basic definitions}

We begin by establishing the basic definitions used throughout this work. 
We first define a circuit architecture as a family of ``blank'' quantum circuits, in which the gate locations are fixed while the specific gate parameters are left unspecified, following the notion introduced in the original RCS proposal~\cite{bouland2019complexity}.

\begin{definition}[Circuit architecture]\label{def: architecture}
The circuit architecture $\mathcal{A}$ is defined as a family of quantum circuits $\{ A^{(k)} \}_{k=1,2,\dots}$.
Each circuit $A^{(k)}$ acts on $k$ qubits and consists of at most $\poly(k)$ one- or two-qubit gates at fixed locations, where specific gate parameters are left unspecified.
A circuit $C$ acting on $n$ qubits over $\mathcal{A}$ is instantiated by taking the $n$-th circuit $A^{(n)}$ and specifying all the gates in the circuit. 
In particular, for each $k$, every qubit in $A^{(k)}$ experiences at least one gate throughout the circuit. 
\end{definition}

%Namely, a circuit architecture can be regarded as the blueprint of a quantum circuit for each qubit number
Thus, a circuit architecture can be viewed as a blueprint for a quantum circuit at each system size, with specific gate instances to be filled in to instantiate the circuit.
We also define the circuit depth as the number of parallel layers of gates, as formalized below.

\begin{definition}[Circuit depth]\label{def: circuit depth}
The circuit depth $d$ is defined as the number of unit-depth layers in the circuit, where each unit-depth layer corresponds to a parallel application of gates with no sequential operations within the same layer.
\end{definition}

The circuit $C$ of depth $d$ can therefore be written as the product $C = U_{d}U_{d-1}\cdots U_{1}$, where $U_{i}$ denotes the parallel application of gates in the $i$th layer of $C$. 
In each unit-depth layer $U_{i}$, every qubit experiences at most one gate operation, possibly no gate operation (identity).

Based on the above definitions of circuit architecture and circuit depth, we define the depth of a circuit architecture as follows.

\begin{definition}[Depth of a circuit architecture]\label{def: circuit depth of architecture}
Let $\mathcal{A} = \{A^{(k)}\}_{k=1,2,\dots}$ be a circuit architecture, where each $A^{(k)}$ represents a $k$-qubit circuit with unspecified gate parameters.
We define $\mathsf{depth}_{\mathcal{A}}(n)$ as the circuit depth of $A^{(n)}$.
%We assume throughout that $\mathsf{depth}_{\mathcal{A}}(n)$ is non-decreasing in $n$ (i.e., $\mathsf{depth}_{\mathcal{A}}(k_1) \leq \mathsf{depth}_{\mathcal{A}}(k_2)$ for $k_1 \leq k_2$) and satisfies $\mathsf{depth}_{\mathcal{A}}(n)=\Omega(1)$.
\end{definition}

For example, for a logarithmic-depth circuit architecture $\mathcal{A}$, one has $\mathsf{depth}_{\mathcal{A}}(n) = \Theta(\log n)$.

Next, we define the random-circuit ensemble used for RCS.
Among random circuit ensembles, we are particularly interested in local Haar-random quantum circuits, in which each gate in the circuit architecture is drawn independently from the Haar measure, as formalized below.

\begin{definition}[Local Haar-random circuit distribution]\label{def: random circuit distribution}
%Let $\mathcal{A}$ be a circuit architecture and let the gates in $\mathcal{A}$ be $\{G_{i}\}_{i=1,2,\dots,m}$. 
We define $\mathcal{H}_{\mathcal{A}}$ as the circuit ensemble over the architecture $\mathcal{A}$ obtained by independently drawing each one- or two-qubit gate from the Haar measure on $\mathrm{U}(2)$ or $\mathrm{U}(4)$, respectively.
\end{definition}

We use local Haar-random circuit ensembles because they provide the standard
setting for the complexity-theoretic hardness conjectures of ideal
RCS~\cite{bouland2019complexity, bouland2022noise, kondo2022quantum, movassagh2023hardness, krovi2022average, bouland2025exponential} and, at the same time, possess the Pauli invariance needed in our noisy-RCS reductions. 
In particular, Pauli invariance allows Pauli errors arising from the depolarizing noise expansion to be absorbed into the random gates without changing the circuit ensemble, and also justifies fixing the output string to $0^n$.
Consequently, our arguments are not tied to the Haar ensemble itself; they extend to any random-circuit ensemble for which analogous ideal-RCS hardness evidence holds and the required Pauli invariance is available.

%We focus on Haar-random circuits because they exhibit Pauli invariance, and the conjectures underlying the classical hardness of ideal RCS are formulated for this ensemble~\cite{bouland2019complexity, bouland2022noise, kondo2022quantum, movassagh2023hardness, krovi2022average, bouland2025exponential}, both of which are essential to our hardness analysis.
%This implies that our results can readily be generalized to broader classes of random-circuit ensembles, as long as they exhibit Pauli invariance and allow analogous hardness evidence in the ideal setting.

Finally, we consider local depolarizing noise, a standard benchmark
model in the complexity-theoretic study of noisy RCS~\cite{aharonov1996limitations, gao2018efficient, boixo2018characterizing, deshpande2022tight, dalzell2024random, boixo2017fourier, fefferman2024effect, li2023entanglement, zhang2022entanglement, zhang2022noise, cheng2021simulating, aharonov2023polynomial, cheng2023efficient,chen2024optimized}.
Its unitality and Pauli-mixture representation make it especially well-suited to analyzing how accumulated noise affects the hardness of RCS: they lead to polynomial dependence on the noise strength and the monotonicity property used in our reductions.
While other noise models, including nonunital noise, may exhibit qualitatively different behavior~\cite{fefferman2024effect}, the depolarizing setting provides a clean and widely used framework for identifying the noise strength
boundary studied in this work. 
We now define the corresponding noisy circuit.

%Finally, we consider the depolarizing noise, which has been widely accepted as a physically relevant noise model in prior studies~\cite{aharonov1996limitations, gao2018efficient, boixo2018characterizing, deshpande2022tight, dalzell2024random, boixo2017fourier, fefferman2024effect, li2023entanglement, zhang2022entanglement, zhang2022noise, cheng2021simulating, aharonov2023polynomial, cheng2023efficient,chen2024optimized}.
%We formally define a noisy quantum circuit as one in which uniform depolarizing noise is applied to each qubit at each layer of the circuit, as formalized below.

\begin{figure*}[t]
\includegraphics[width=0.87\linewidth]{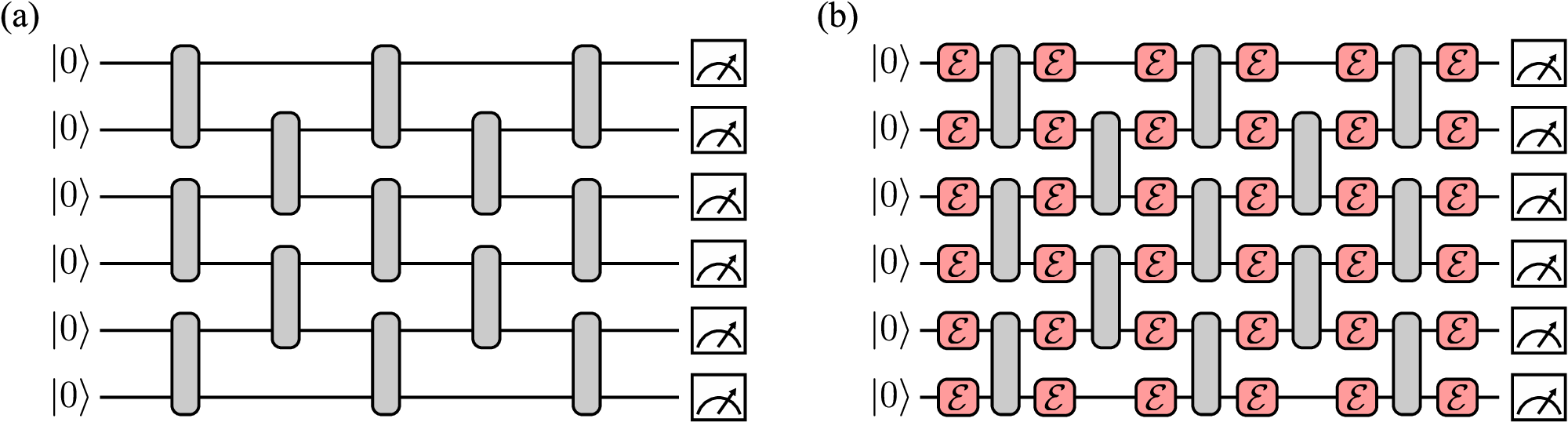}
\caption{Schematics of our RCS settings. 
(a) Schematic of ideal RCS. 
The gray boxes denote Haar-random gates, so that the overall circuit is drawn from the ensemble $\mathcal{H}_{\mathcal{A}}$ in Definition~\ref{def: random circuit distribution}. 
The input $\ket{0^{n}}$ is evolved by a random circuit $C \sim \mathcal{H}_{\mathcal{A}}$ and measured in the computational basis, yielding an output $x \in \{0,1\}^n$ distributed according to $p(C,x)$ in Eq.~\eqref{eq: ideal output probability}. 
(b) Schematic of noisy RCS. 
From (a), single-qubit depolarizing channels (red boxes) in Definition~\ref{def: noisy quantum circuit} are applied to each qubit at each time step (depth), such that the output $x \in \{0,1\}^n$ is distributed according to $\widetilde{p}(C, \gamma, x)$ in Eq.~\eqref{eq: noisy output probability}.
While we draw a 1D local architecture in the figure for illustration, we consider general circuit architecture $\mathcal{A}$ throughout this work. 
%In each figure, gray boxes denote Haar-random gates, so that the overall circuit is drawn from the ensemble  in Definition~\ref{def: random circuit distribution}. 
%Also, red boxes in (b) denote the single-qubit depolarizing channel from Definition~\ref{def: noisy quantum circuit}.
}
\label{fig: ideal_noisy_RCS}
\end{figure*}

\begin{definition}[Noisy quantum circuit]\label{def: noisy quantum circuit}
We define a noisy quantum circuit $\widetilde{C}$ as a quantum circuit $C$ in which single-qubit depolarizing noise $\mathcal{E}_{\gamma}(\rho) = (1-\gamma)\rho + \gamma\frac{I}{2}\Tr(\rho)$ with uniform noise strength $\gamma \in [0,1]$ is applied independently to each qubit at each time step in the circuit $C$.
%More specifically, for the circuit $C = U_{d}U_{d-1}\cdots U_{1}$ with depth $d$, the noisy circuit $\widetilde{C}$ first applies noise channel $\mathcal{E}(\cdot)$ to all qubits, followed by $U_1$, then applies noise again, followed by $U_2$, and so on, up to the final layer $U_{d}$.
\end{definition}

The noisy circuit considered throughout this work is illustrated in Fig.~\ref{fig: ideal_noisy_RCS}.

\subsection{Output probability distribution}

Let $C = U_{d}U_{d-1}\cdots U_{1}$ be an $n$-qubit circuit of depth $d$, where $U_{i}$ denotes the $i$th layer of $C$.
For the circuit $C$ and output string $x \in \{0,1\}^{n}$, we define the \textit{ideal} output probability as
\begin{align}\label{eq: ideal output probability}
    p(C,x) \coloneqq |\bra{x}C\ket{0^{n}}|^2 ,
\end{align}
where $\ket{x}$ denotes the computational basis for $x \in \{0,1\}^{n}$.
Similarly, let $\widetilde{p}(C,\gamma, x)$ be the \textit{noisy} output probability of observing $x \in \{ 0, 1\}^{n}$, corresponding to the noisy circuit $\widetilde{C}$ defined in Definition~\ref{def: noisy quantum circuit}. 
Specifically, writing $\mathcal{U}^{(i)}$ as the unitary channel corresponding to the $i$th layer $U_{i}$ of $C$ such that $\mathcal{U}^{(i)}(\rho)=U_i\rho\,U_{i}^{\dagger}$, the channel $\mathcal{N}_{C, \gamma}$ corresponding to the noisy circuit $\widetilde{C}$ can be written as
\begin{align}
\label{eq:noisy-channel}
\mathcal{N}_{C, \gamma}
:=
\bigl(\mathcal{E}_\gamma^{\otimes n}\bigr)\circ \mathcal{U}^{(d)}
\circ \cdots \circ
\bigl(\mathcal{E}_\gamma^{\otimes n}\bigr)\circ \mathcal{U}^{(1)} 
\circ
\bigl(\mathcal{E}_\gamma^{\otimes n}\bigr)   .
\end{align}
Here, the single-qubit depolarizing noise channel $\mathcal{E}_{\gamma}$ admits the Pauli-mixture representation~\cite{dalzell2024random} 
\begin{align}
    \mathcal{E}_{\gamma}(\rho) 
    &= 
    \left(1-\frac{3}{4}\gamma\right)\rho + \frac{1}{4}\gamma\sum_{P \in \{X, Y, Z\}} P\rho P . \label{def: noise channel}
\end{align}
Substituting Eq.~\eqref{def: noise channel} into Eq.~\eqref{eq:noisy-channel}, the corresponding noisy output probability distribution can be written as
\begin{align}
     \widetilde{p}(C,\gamma,x) 
     &= \Tr\left[  \ket{x}\!\bra{x} \; \mathcal{N}_{C, \gamma} \bigl( \,\ket{0^n}\!\bra{0^n} \,\bigr)  \right] \\
     &= \sum_{s \in \mathsf{P}_n^{d+1}} \left(1-\frac{3}{4}\gamma\right)^{N - |s|} \left( \frac{1}{4}\gamma \right)^{|s|} p(C_{s},x)  , \label{eq: noisy output probability}
\end{align}
where $s = (s_0, \dots, s_d) \in \mathsf{P}_{n}^{d+1}$ is a \textit{Pauli path} with each $s_i$ belonging to the $n$-qubit Pauli operator set
\begin{align}
    s_{i} \in \mathsf{P}_{n} \coloneqq \left\{ I, X, Y, Z  \right\}^{\otimes n} ,
\end{align}
and $|s|$ denotes the \textit{Hamming weight} of $s$, i.e., the number of non-identity Pauli operators in $s$. 
Also, $N = n(d+1)$ is the total number of depolarizing ``sites" (the number of depolarizing noise channels) in the noisy circuit. 
For a Pauli path $s = (s_0, \dots, s_d)$, $C_{s}$ is the circuit obtained from $C$ by inserting the $n$-qubit Pauli operation $s_i$ after the $i$th layer $U_{i}$ for $i = 1,\dots,d$, as well as applying $s_{0}$ before $U_{1}$. 
In particular, $C_{s_{\rm id}}=C$ for the all-identity Pauli path $s_{\rm id}$ (so $|s_{\rm id}| = 0$), and one can readily check $\widetilde{p}(C,\gamma,x) = p(C,x)$ in the noiseless limit $\gamma = 0$, since only the all-identity path contributes to the sum in Eq.~\eqref{eq: noisy output probability}.

%Following our convention in Definition~\ref{def: noisy quantum circuit}, when this noise channel is applied independently to each qubit at each time step of the circuit $C$,

\iffalse
Throughout this work, we consider the random circuit $C$ drawn from the circuit ensemble $\mathcal{H}_{\mathcal{A}}$ in Definition~\ref{def: random circuit distribution}.
By our convention for the circuit architecture $\mathcal{A}$ in Definition~\ref{def: architecture}, every qubit in any circuit drawn from $\mathcal{H}_{\mathcal{A}}$ undergoes at least one Haar-random gate.
Therefore, by Haar invariance, $p(C,x)$ (and likewise $\widetilde{p}(C,\gamma, x)$) follows essentially the same distribution over $C \sim \mathcal{H}_{\mathcal{A}}$ for every $x \in \{ 0, 1\}^{n}$, because Pauli-$X$ operations before measurement can be absorbed into the Haar-random gates without changing the circuit distribution.
Hence, for simplicity, we fix the output to $x = 0^{n}$ within our hardness arguments, and employ the convention without explicit dependence on $x$ such that
\fi

Throughout this work, we take $C \sim \mathcal{H}_{\mathcal{A}}$ in Definition~\ref{def: random circuit distribution}.
For Haar-random circuit ensembles, the random variables $p(C,x)$ and $\widetilde p(C,\gamma,x)$ have the same distribution over $C\sim\mathcal H_\mathcal A$ for every fixed output string $x$: Pauli-$X$ operations that map $0^n$ to any $x \in \{0,1\}^n$ can be absorbed into the Haar-random gates, and the depolarizing channel is Pauli-covariant.
We therefore fix $x=0^n$ throughout the hardness arguments and write
\begin{align}\label{eq: def of p(C)}
\begin{split}
    p(C) \coloneqq p(C,0^{n}),
    \quad
    \widetilde{p}(C,\gamma) \coloneqq \widetilde{p}(C,\gamma, 0^{n})  ,
    \end{split}
\end{align}
for $p(C,x)$ and $\widetilde{p}(C,\gamma, x)$ given in Eq.~\eqref{eq: ideal output probability} and Eq.~\eqref{eq: noisy output probability}, respectively.

\section{Main problems and results}\label{section: main problem and result}

This section presents our main classical-simulation hardness results for noisy RCS, where the overall argument is outlined in Fig.~\ref{fig: outline}.
%Throughout, by ``simulation" we mean the classical sampling problem that given as input a description of an $n$-qubit circuit $C \sim \mathcal{H}_{\mathcal{A}}$ and error parameter $\beta$, with implicitly given architecture $\mathcal{A}$ and noise strength $\gamma$, outputs a string $x \in \{0,1\}^{n}$ according to a distribution that approximates the noisy output distribution $\widetilde{p}(C,\gamma,x)$ in Eq.~\eqref{eq: noisy output probability} within TVD $\beta$ in time $\poly(n, 1/\beta)$, following the standard definition in the literature~\cite{bouland2019complexity, movassagh2023hardness, bouland2022noise}.
We begin by reviewing prior analyses on ideal RCS, highlighting the key requirements in proving its hardness.
Building on this foundation, we then introduce our main hardness results for noisy RCS and identify the noise-strength thresholds below which noisy RCS retains the hardness of ideal RCS.

%Lastly, we clarify how these results can be used to ultimately establish the classical intractability of noisy RCS in these noise regimes.

\subsection{Hardness framework for ideal RCS}

We start with hardness arguments for ideal RCS in~\cite{bouland2019complexity, bouland2022noise, kondo2022quantum, movassagh2023hardness, krovi2022average}. 
We first formalize the computational task underlying the hardness of ideal RCS, namely, the \textit{average-case} estimation of ideal output probability $p(C)$ in Eq.~\eqref{eq: def of p(C)} over a randomly chosen circuit $C \sim \mathcal{H}_{\mathcal{A}}$, which we refer to as $\mathcal{A}$-\textsc{Ideal-Probability-Estimation}.

\begin{definition}[Average-case ideal probability estimation]\label{def: ideal probability estimation problem}
    The $\mathcal{A}$-\textsc{Ideal-Probability-Estimation} task is defined as follows. 
    Given an $n$-qubit random circuit $C \sim \mathcal{H}_{\mathcal{A}}$ for an implicit architecture $\mathcal{A}$, together with error parameters $\varepsilon_0, \delta_0 > 0$, the task is to output an estimate of the ideal output probability $p(C)$ within additive error $\pm \varepsilon_0 2^{-n}$, with probability at least $1 -\delta_0$ over the choice of $C$, in $\poly(n,\varepsilon_0^{-1},\delta_0^{-1})$ time.
\end{definition}

\begin{figure*}[t]
\includegraphics[width=0.72\linewidth]{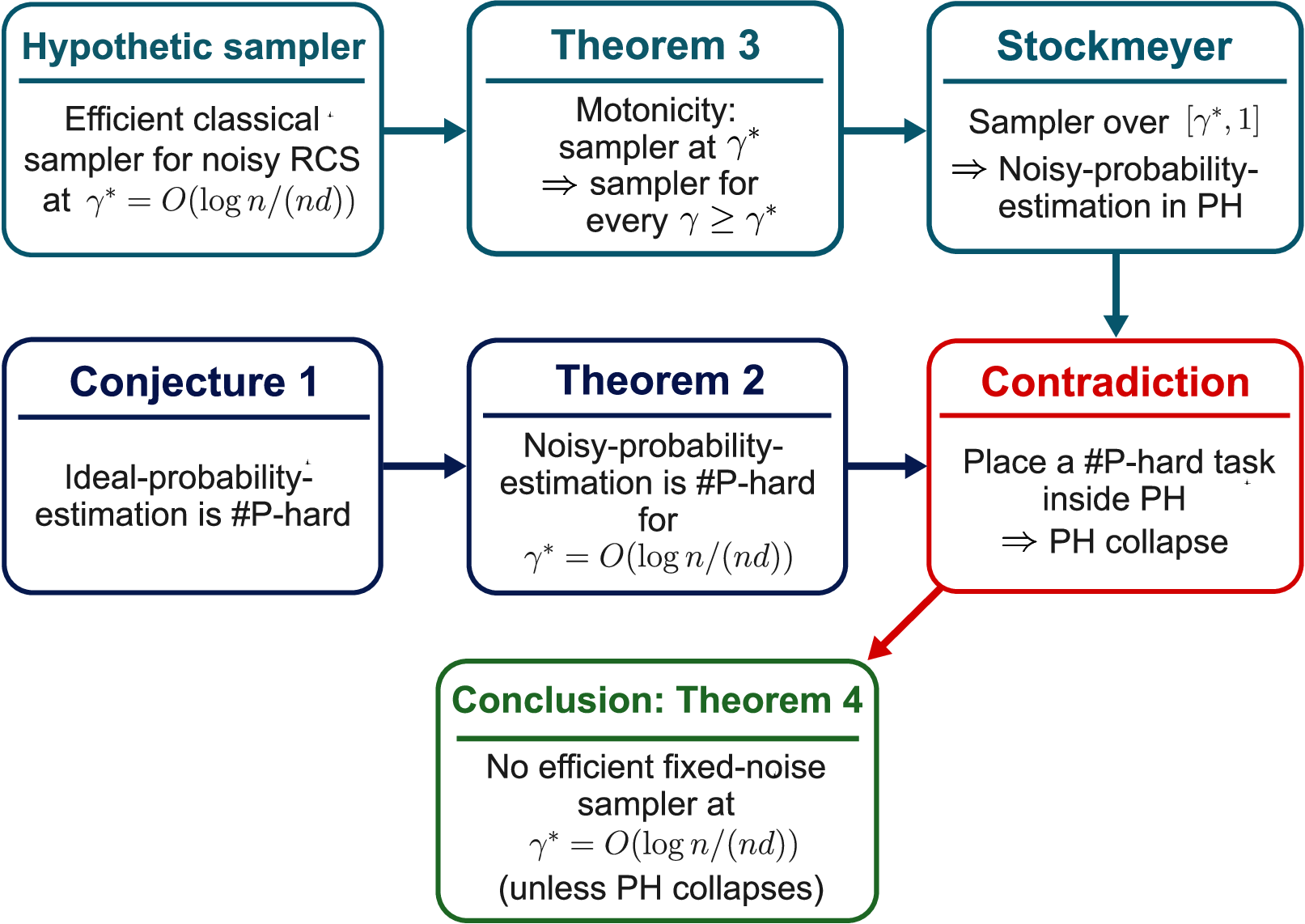}
\caption{
Logical flow of the noisy-RCS sampling-hardness proof. 
In the top row, assuming an efficient classical sampler for noisy RCS at noise strength $\gamma^* = O(\log n/(nd))$, Theorem~\ref{thm: complexity decrease} implies an efficient sampler for every larger noise strength $\gamma \geq \gamma^*$. 
Stockmeyer's reduction (Lemma~\ref{lemma: stockmeyer}) then yields noisy-probability estimation for any $\gamma \geq \gamma^*$ within a finite level of the polynomial hierarchy (PH).
In the bottom row, Conjecture~\ref{conj: average-case hardness} gives average-case \#P-hardness of ideal-probability estimation, and Theorem~\ref{thm 1} transfers this hardness to noisy-probability estimation when $\gamma \geq \gamma^*$ is given as an input parameter. 
Combining the two arguments would place a \#P-hard estimation task inside a finite level of PH, implying a collapse of PH.
Therefore, under Conjecture~\ref{conj: average-case hardness} and the non-collapse of PH, no efficient classical sampler exists for noisy RCS at noise strength $\gamma^* = O(\log n/(nd))$, yielding Theorem~\ref{corol: hardness}. 
%Outline of our hardness analysis on noisy RCS. 
%To sum up, starting from Conjecture~\ref{conj: average-case hardness}, we derive classical simulation hardness of noisy RCS (our main conclusion) at the end, as depicted in Theorem~\ref{corol: hardness}. 
%Here, the proof of Theorem~\ref{thm 1} is given in Sec.~\ref{section: proof of theorem 1}, the explicit construction of the architecture $\mathcal{A}_{L}^{\rm br}$ and the proof of Theorem~\ref{thm 2} is given in Sec.~\ref{section: proof of theorem 2}, and the proof of Theorem~\ref{thm: complexity decrease} is given in Sec.~\ref{section: proof of theorem 3}.
%We also present numerical analysis of anti-concentration in Sec.~\ref{section: numerical evidence}, which supports the anti-concentration assumption required by Theorem~\ref{thm 2}.
%Within this outline, we implicitly assume that PH does not collapse to a finite level. 
}
\label{fig: outline}
\end{figure*}

Note that the computational complexity of $\mathcal{A}$-\textsc{Ideal-Probability-Estimation} in Definition~\ref{def: ideal probability estimation problem} depends crucially on the architecture $\mathcal{A}$; 
for example, estimating $p(C)$ is easy for trivial architectures. 
We assume that there exists a circuit architecture $\mathcal{A}_0$ such that the $\mathcal{A}$-\textsc{Ideal-Probability-Estimation} task is \#P-hard.

\begin{conjecture}[\#P-hardness of average-case ideal probability estimation]\label{conj: average-case hardness}
There exists a circuit architecture $\mathcal{A}_0$ such that the $\mathcal{A}_0$-\textsc{Ideal-Probability-Estimation} in Definition~\ref{def: ideal probability estimation problem} is \#P-hard.
\end{conjecture}

%Under Conjecture~\ref{conj: average-case hardness}, by the complexity-theoretic reduction based on Stockmeyer's approximate counting  algorithm~\cite{stockmeyer1985approximation} established in the preceding arguments~\cite{aaronson2011computational, bouland2019complexity, bouland2022noise, kondo2022quantum, movassagh2023hardness, krovi2022average, bouland2025exponential}, the classical simulation of ideal RCS within an arbitrary inverse-polynomial TVD over the architecture $\mathcal{A}_0$ in Conjecture~\ref{conj: average-case hardness} is intractable, unless the polynomial hierarchy (PH) collapses.

Under Conjecture~\ref{conj: average-case hardness}, the reduction based on Stockmeyer's algorithm~\cite{stockmeyer1985approximation} established in Ref.~\cite{aaronson2011computational, bouland2019complexity, bouland2022noise, kondo2022quantum, movassagh2023hardness, krovi2022average, bouland2025exponential} implies that classical simulation of ideal RCS over $\mathcal{A}_0$ in Conjecture~\ref{conj: average-case hardness} (within any inverse-polynomial TVD) is intractable unless the polynomial hierarchy (PH) collapses.

\begin{remark}
    Reference~\cite[Theorem~3]{deshpande2022tight} shows that generic parallel architectures of sub-logarithmic depth $o(\log n)$ lack anti-concentration, in the sense that most output probabilities can be well approximated by ``$0$" to the relevant additive accuracy.
    This result rules out Conjecture~\ref{conj: average-case hardness} for such sublogarithmic-depth parallel architectures; see also Ref.~\cite{go2024exploring}.
    Outside this regime, the conjecture is not precluded by this particular obstruction, though it remains an assumption.
\end{remark}

Before proceeding, we briefly summarize the progress to date toward Conjecture~\ref{conj: average-case hardness}.
Reference~\cite{bouland2019complexity} established the \#P-hardness of $\mathcal{A}$-\textsc{Ideal-Probability-Estimation} in the setting of exact computation ($\varepsilon_0 = 0$). 
%for the architecture $\mathcal{A}$ exhibiting the worst-case hardness,
Subsequently, Ref.~\cite{movassagh2023hardness} strengthened this result by improving the tolerated additive-error scale to $\varepsilon_0 = e^{-O(m^3)}$, where $m$ denotes the number of gates in an $n$-qubit circuit. 
Further robustness was obtained in Refs.~\cite{bouland2017complexity, kondo2022quantum}, achieving the imprecision level $\varepsilon_0 = e^{-O(m\log m)}$.
This was further improved in Ref.~\cite{krovi2022average} to $\varepsilon_0 = e^{-O(m)}$.
Most recently, the RCS corollary of Ref.~\cite{bouland2025exponential} achieved $\varepsilon_0 = 2^{-O(n^{\beta})}$ for a fixed $\beta >0$.
However, the conjecture for inverse-polynomial $\varepsilon_0$ still remains open.

\subsection{Average-case \#P-hardness of noisy output probability estimation}

We now state our main results.
In particular, in direct analogy with Definition~\ref{def: ideal probability estimation problem}, we formally define the computational task that underlies our noisy-RCS hardness result, namely, average-case estimation of noisy output probability $\widetilde{p}(C, \gamma)$ in Eq.~\eqref{eq: def of p(C)}, which we refer to as $(\mathcal{A}, \gamma^*)$-\textsc{Noisy-Probability-Estimation}.

\begin{definition}[Average-case noisy probability estimation]\label{def: noisy probability estimation problem}
    The $(\mathcal{A},  \gamma^*)$-\textsc{Noisy-Probability-Estimation} task is defined as follows. 
    Given an $n$-qubit random circuit $C \sim \mathcal{H}_{\mathcal{A}}$ and noise parameter $\gamma \in [\gamma^*,1]$ for an implicit architecture $\mathcal{A}$ and noise strength $\gamma^*$, together with error parameters $\varepsilon, \delta > 0$, the task is to output an estimate of the noisy output probability $\widetilde{p}(C, \gamma)$ within additive error $\pm \varepsilon 2^{-n}$, with probability at least $1 -\delta$ over the choice of $C$, in $\poly(n,\varepsilon^{-1},\delta^{-1})$ time.
\end{definition}

%Notice that compared to its ideal counterpart in Definition~\ref{def: ideal probability estimation problem}, the $(\mathcal{A},  \gamma^*)$-\textsc{Noisy-Probability-Estimation} task now takes the noise parameter $\gamma \in [\gamma^*,1]$ additionally. 

Compared with its ideal counterpart in Definition~\ref{def: ideal probability estimation problem}, the $(\mathcal{A},\gamma^*)$-\textsc{Noisy-Probability-Estimation} task additionally takes a noise parameter $\gamma\in[\gamma^*,1]$ as input.
Here, $\gamma^*$ is a \textit{fixed}, lower-noise parameter for the estimation problem, possibly depending on the system size $n$, whereas $\gamma$ is provided as input.
Allowing $\gamma$ to be queried as an input parameter plays a central role in establishing the hardness of noisy RCS, analogous to recent approaches used to prove hardness results for noisy BosonSampling~\cite{go2025quantum, go2025sufficient}.
We also note that  as in the ideal RCS setting, \#P-hardness of the $(\mathcal{A},  \gamma^*)$-\textsc{Noisy-Probability-Estimation} implies classical intractability of a noisy-RCS sampling task in which the noise strength is supplied as an input from $[\gamma^*,1]$; see Lemma~\ref{lemma: stockmeyer} discussed later.

Based on this definition, we identify a threshold of $\gamma^*$ such that the $(\mathcal{A}_0,  \gamma^*)$-\textsc{Noisy-Probability-Estimation} task for the architecture $\mathcal{A}_0$ in Conjecture~\ref{conj: average-case hardness} is \#P-hard under Conjecture~\ref{conj: average-case hardness}.

%and for any $\gamma^* \leq \gamma_{\rm th}$,

\begin{theorem}[\#P-hardness of average-case noisy probability estimation]\label{thm 1}
    Let $\mathcal{A}_0$ be a circuit architecture in Conjecture~\ref{conj: average-case hardness} and let $d = \mathsf{depth}_{\mathcal{A}_0}(n)$.
    If
    \begin{align}\label{eq: thm: noise threshold 1}
        \gamma^* = O\left( \frac{\log n}{nd} \right)   ,
    \end{align}
    then $(\mathcal{A}_0,  \gamma^*)$-\textsc{Noisy-Probability-Estimation} in Definition~\ref{def: noisy probability estimation problem} is \#P-hard under Conjecture~\ref{conj: average-case hardness}.
\end{theorem}

\iffalse
\begin{theorem}\label{thm 1}
    Let $\mathcal{A}_0$ be a circuit architecture in Conjecture~\ref{conj: average-case hardness} and let $d = \mathsf{depth}_{\mathcal{A}_0}(n)$.
    %Suppose that the probability distribution $p(C)$ over circuits $C \sim \mathcal{H}_{\mathcal{A}_0}$ satisfies the weak anti-concentration condition in Definition~\ref{def: weak AC}. 
    Then, there exists a noise threshold $\gamma_{\rm th}$ given by 
    \begin{align}\label{eq: thm: noise threshold 1}
        \gamma_{\rm th} = \Theta\left( \frac{\log n}{nd} \right)   ,
    \end{align}
    such that $(\mathcal{A}_0,  \gamma^*)$-\textsc{Noisy-Probability-Estimation} in Definition~\ref{def: noisy probability estimation problem} with any $\gamma^* \leq \gamma_{\rm th}$ is \#P-hard under Conjecture~\ref{conj: average-case hardness}.
\end{theorem}
\fi

We prove Theorem~\ref{thm 1} in Sec.~\ref{section: proof of theorem 1}. 
The proof establishes a complexity-theoretic reduction from $\mathcal{A}_{0}$-\textsc{Ideal-Probability-Estimation} to $(\mathcal{A}_{0},\gamma^*)$-\textsc{Noisy-Probability-Estimation}, and shows that the condition on $\gamma^*$ in Eq.~\eqref{eq: thm: noise threshold 1} ensures that the reduction has only polynomial overhead.
Thus, under Conjecture~\ref{conj: average-case hardness}, this yields the \#P-hardness claimed in Theorem~\ref{thm 1}.

\subsection{Hardness of classically simulating noisy RCS}

In the preceding arguments, we identified the boundary of $\gamma^*$ such that the $(\mathcal{A}_0,  \gamma^*)$-\textsc{Noisy-Probability-Estimation} task is \#P-hard under Conjecture~\ref{conj: average-case hardness}. 
By the standard reduction based on Stockmeyer’s algorithm~\cite{stockmeyer1985approximation}, this estimation task can be solved within a finite level of the PH given oracle access to an approximate sampler for noisy RCS over $\mathcal{H}_{\mathcal{A}_0}$ within inverse-polynomial TVD and that takes any noise strength $\gamma\in[\gamma^*,1]$ as an input (see Refs.~\cite{aaronson2011computational, bouland2019complexity, bouland2022noise, kondo2022quantum, movassagh2023hardness, krovi2022average, bouland2023complexity, bouland2025exponential} for detailed reduction procedures).
Hence, together with Toda's theorem $\rm{PH} \subseteq {\rm P^{\# P}}$~\cite{toda1991pp}, this gives the following lemma.

%By the standard reduction established in prior studies~\cite{aaronson2011computational, bouland2019complexity, bouland2022noise, kondo2022quantum, movassagh2023hardness, krovi2022average, bouland2023complexity, bouland2025exponential}, based on Stockmeyer’s algorithm~\cite{stockmeyer1985approximation}, oracle access to an approximate sampler that has bounded TVD enables estimation of most output probabilities with bounded error parameters within a finite level of the PH. 

\begin{lemma}[Stockmeyer's reduction]\label{lemma: stockmeyer}
    Suppose that $(\mathcal{A},  \gamma^*)$-\textsc{Noisy-Probability-Estimation} is \#P-hard for certain $(\mathcal{A},  \gamma^*)$.
    If there exists a polynomial-time classical sampler that, given a noise strength $\gamma\in[\gamma^*,1]$ as input, approximately samples from noisy RCS over $\mathcal{H}_{\mathcal{A}}$ within inverse-polynomial TVD, then PH collapses to a finite level.
    %Then, the existence of a polynomial-time classical simulator that approximately simulates noisy RCS over $\mathcal{H}_\mathcal{A}$ within inverse-polynomial TVD, and that is allowed to take the noise strength $\gamma \in [\gamma^*,1]$ as an input, implies a collapse of the PH to a finite level.
\end{lemma}

Therefore, Lemma~\ref{lemma: stockmeyer} implies that for $\gamma^*$ given in Theorem~\ref{thm 1}, noisy RCS over $\mathcal{H}_{\mathcal{A}_0}$ is classically intractable when the noise strength $\gamma\in[\gamma^*,1]$ is supplied as an input, under Conjecture~\ref{conj: average-case hardness} and non-collapse of PH. 
%In other words, there exists at least one $\gamma \in [\gamma^*,1]$ for which simulating noisy RCS with noise strength $\gamma$ is classically hard. 
However, this implication alone does not establish classical-simulation hardness for noisy RCS in the standard fixed-noise formulation, in which the noise strength is specified in advance as the target value $\gamma^*$~\cite{aharonov1996limitations, gao2018efficient, deshpande2022tight, aharonov2023polynomial,  nelson2025limitations, nelson2026polynomial, lee2025classical, zhang2025classically, noh2020efficient, zhang2023noisy, chen2018classical, cirstoiu2024fourier, huang2020classical, hangleiter2023computational}.

To bridge this gap, an additional ingredient is required, namely, a guarantee that increasing the depolarizing noise strength cannot make the sampling problem harder.
In fact, this is a natural expectation, supported by an extensive body of evidence that stronger noise makes quantum systems easier to simulate classically~\cite{renema2018efficient,  moylett2019classically, shchesnovich2019noise, garcia2019simulating, oszmaniec2018classical,  brod2020classical,  oh2021classical, oh2023classical, oh2025classical, noh2020efficient, aharonov2023polynomial, gao2018efficient, bremner2017achieving, rajakumar2025polynomial, aharonov1996limitations, lee2025classical, muller2016relative, nelson2025limitations, oh2025recent, deshpande2022tight, nelson2026polynomial, dalzell2024random, zhang2025classically}.
%, both for sampling from their output probability distributions and for estimating their output probabilities
%As similarly assumed in Ref.~\cite{go2025quantum, go2025sufficient}, these observations motivate the following physically plausible conjecture. 
Motivated by this intuition, we prove a monotonic decrease in complexity for our noisy RCS setting: simulating noisy RCS with a smaller noise strength implies the ability to simulate noisy RCS with any larger noise strength.

%that such a monotonic decrease in complexity generally holds for our noisy RCS setting, formulated as follows.

\begin{theorem}[Monotonicity of noisy RCS simulation]\label{thm: complexity decrease}
    Fix $\gamma_1\in[0,1]$ and $\beta\in[0,1]$. 
    If noisy RCS for a Haar-random circuit ensemble $\mathcal{H}_{\mathcal{A}}$ under the depolarizing noise model of Definition~\ref{def: noisy quantum circuit} admits a polynomial-time sampler within TVD $\beta$ at noise strength $\gamma_1$, then for every $\gamma_2\in[\gamma_1,1]$, noisy RCS over the same ensemble at noise strength $\gamma_2$ also admits a polynomial-time sampler within TVD $\beta$.
    %Fix $\gamma_{1} \in [0,1]$ and $\beta \in [0,1]$.
    %Suppose there exists a polynomial-time sampler that approximately simulates noisy RCS over a fixed circuit ensemble under depolarizing noise in Definition~\ref{def: noisy quantum circuit} with a fixed noise strength $\gamma_{1}$, and within TVD $\beta$.
    %Then, one can simulate noisy RCS for arbitrarily given noise strength $\gamma_{2} \in [\gamma_{1}, 1]$ over the same circuit ensemble, within the same TVD $\beta$, in polynomial time. 
\end{theorem}

%The computational complexity of estimating output probabilities of noisy RCS within a fixed additive error decreases monotonically with increasing $\gamma$.

We defer the detailed proof of Theorem~\ref{thm: complexity decrease} to Sec.~\ref{section: proof of theorem 3}. 
To sketch the proof, we use the observation that for every $\gamma_2\in[\gamma_1,1]$, the depolarizing channel $\mathcal{E}_{\gamma_2}$ can be decomposed as $\mathcal{E}_{\gamma_1}\circ\mathcal{E}_{\gamma_{\rm add}}$ for some additional noise strength $\gamma_{\rm add}\in[0,1]$.
Hence, given a classical sampler for noise strength $\gamma_1$, one can simulate any larger noise strength $\gamma_2\ge\gamma_1$ by sampling the additional random Pauli errors induced by $\mathcal{E}_{\gamma_{\rm add}}$, absorbing these extra Pauli errors into the circuit, and then running the $\gamma_1$-sampler.
The resulting distribution is a convex mixture of distributions simulated by the $\gamma_1$-sampler; therefore, the total variation error does not increase.

%Finally, by Theorem~\ref{thm: complexity decrease}, if we have an efficient sampler for noisy RCS with fixed noise strength $\gamma^*$, then we also have efficient samplers for any $\gamma \in [\gamma^*,1]$.
%Therefore, combining Theorem~\ref{thm 1}, Lemma~\ref{lemma: stockmeyer}, and Theorem~\ref{thm: complexity decrease}, we arrive at our main conclusion.

Therefore, by Theorem~\ref{thm: complexity decrease}, efficient simulation of noisy RCS at noise strength $\gamma^*$ implies efficient simulation at every larger noise strength $\gamma \in[\gamma^*,1]$. 
Finally, combining this monotonicity result with Theorem~\ref{thm 1} and Lemma~\ref{lemma: stockmeyer}, we arrive at our main conclusion.

%Equivalently, if there exists at least one $\gamma \in [\gamma^*,1]$ such that simulating noisy RCS with noise strength $\gamma$ is classically hard, then noisy RCS with a \textit{fixed} noise strength $\gamma^*$ is also classically hard.

\begin{theorem}[Hardness of noisy RCS]\label{corol: hardness}
    Let $\mathcal{A}_{0}$ be a circuit architecture in Conjecture~\ref{conj: average-case hardness}, and let $d = \mathsf{depth}_{\mathcal{A}_0}(n)$.
    Then, unless PH collapses to a finite level, no polynomial-time classical sampler can approximately simulate noisy RCS over $\mathcal{H}_{\mathcal{A}_{0}}$ within inverse-polynomial TVD at noise strength $\gamma^*$ satisfying
    \begin{align}
    \gamma^* = O\left(\frac{\log n}{nd}\right)  ,
    \end{align}
    under Conjecture~\ref{conj: average-case hardness}.
\end{theorem}

To aid understanding, we summarize the overall structure of the argument in Fig.~\ref{fig: outline}.

\section{Proof of Theorem~\ref{thm 1}}\label{section: proof of theorem 1}

This section proves Theorem~\ref{thm 1}.
To this end, we establish a complexity-theoretic reduction from the $\mathcal{A}_0$-\textsc{Ideal-Probability-Estimation} in Definition~\ref{def: ideal probability estimation problem} to the $(\mathcal{A}_0,  \gamma^*)$-\textsc{Noisy-Probability-Estimation} in Definition~\ref{def: noisy probability estimation problem}.
That is, we show how to estimate the ideal output probability $p(C)$ using estimates of the noisy output probability $\widetilde{p}(C,\gamma)$ for $\gamma \in [\gamma^*, 1]$.
We then identify the boundary of the noise strength $\gamma^*$ for which this reduction has only polynomial overhead.
Finally, under Conjecture~\ref{conj: average-case hardness} that the $\mathcal{A}_0$-\textsc{Ideal-Probability-Estimation} is \#P-hard, it follows that the $(\mathcal{A}_0,  \gamma^*)$-\textsc{Noisy-Probability-Estimation} remains \#P-hard, as claimed in Theorem~\ref{thm 1}. 
%\cor{(provided anticoncentration and \#P-hardness for )}

%following an approach similar to that used in hardness analyses of noisy boson sampling~\cite{aaronson2016bosonsampling, go2025quantum, go2025sufficient}.

To proceed, consider an $n$-qubit circuit $C \sim \mathcal{H}_{\mathcal{A}_0}$ with depth $d = \mathsf{depth}_{\mathcal{A}_0}(n)$. 
Note that for the ideal output probability $p(C) = p(C,0^{n})$ in Eq.~\eqref{eq: def of p(C)}, the noisy output probability $\widetilde{p}(C,\gamma)$ in Eq.~\eqref{eq: def of p(C)} can be expressed as
\begin{align}\label{eq: noisy output probability 2}
    \widetilde{p}(C,\gamma) = \sum_{k=0}^{N}  \left(1-\frac{3}{4}\gamma\right)^{N -k} \left( \frac{1}{4}\gamma \right)^{k} \sum_{\substack{s \in \mathsf{P}_{n}^{d+1} \\ |s| = k}}  p(C_{s})   ,
\end{align}
where $N=n(d+1)$ is the number of depolarizing sites, equivalently, the maximum possible Hamming weight of a Pauli path, and $C_{s}$ is the circuit obtained from $C$ by inserting Pauli operators according to $s$.
Crucially, $\widetilde{p}(C,\gamma)$ is a polynomial in $\gamma$ of degree at most $N$, and converges to its ideal counterpart $p(C)$ in the noiseless limit $\gamma \rightarrow 0$.

We then construct a low-degree polynomial in $\gamma$ that approximates the noisy output probability $\widetilde{p}(C,\gamma)$. 
Specifically, for $l < N$, we define the $l$-degree approximation $\widetilde{p}_{l}(C,\gamma)$ as
\begin{align}\label{eq: low degree approx for noisy prob}
    \widetilde{p}_{l}(C,\gamma) \;\coloneqq\;  \sum_{k=0}^{N} g_{k}^{(l)}(\gamma) \sum_{\substack{s \in \mathsf{P}_{n}^{d+1} \\ |s| = k}} p(C_s)  , 
\end{align}
where $g_{k}^{(l)}(\gamma)$ for each $k$ is a yet-to-be-determined polynomial in $\gamma$ of allowed degree at most $l$ that approximates 
\begin{align}
    g_{k}^{(l)}(\gamma) \;\approx  \;  \left(1-\frac{3}{4}\gamma\right)^{N -k} \left( \frac{1}{4}\gamma \right)^{k}  ,
\end{align}
and satisfies $g_{k}^{(l)}(0) = \delta_{k0}$ (i.e., $1$ for $k=0$, and $0$ otherwise), thus ensuring that $\widetilde{p}_{l}(C,0) = \widetilde{p}(C,0) = p(C)$.

In the reduction, we will query the oracle only at noise values in a smaller interval $[\gamma^*,\gamma_{\max}]$, where $\gamma_{\max}>\gamma^*$ will be chosen later.
%For the input noise parameter $\gamma \in [\gamma^*,1]$ of $(\mathcal{A}_0,  \gamma^*)$-\textsc{Noisy-Probability-Estimation} in Definition~\ref{def: noisy probability estimation problem}, we temporarily set the maximum value $\gamma_{\rm max} > \gamma^*$ such that the input $\gamma$ is bounded in the interval $\gamma \in [\gamma^*, \gamma_{\rm max}]$, where $\gamma_{\rm max}$ is left unspecified at this moment.
For query values restricted to this interval, we construct a family of polynomials $\{g_{k}^{(l)}(\gamma)\}_{k}$ such that the approximation $ \widetilde{p}_{l}(C,\gamma) $ is close to $ \widetilde{p}(C,\gamma) $ over a large fraction of $C \sim \mathcal{H}_{\mathcal{A}_0}$, with the approximation error bounded in terms of the maximum input $\gamma_{\rm max}$.

%approximate $\{(1-\gamma)^{k}\}_{k}$, 

\begin{lemma}[Low-degree approximation of $\widetilde{p}(C,\gamma)$]\label{lemma: low-degree approximation}
For any $0<\delta\le1$, any $\gamma_{\max}\in(0,1]$, and any
integer $1\le l<N$, there exists a family of polynomials $\{g_{k}^{(l)}(\gamma)\}_{k=0}^{N}$ such that each polynomial $g_{k}^{(l)}(\gamma)$ has degree at most $l$ and satisfies $g_{k}^{(l)}(0) = \delta_{k0}$, while making the approximation $\widetilde{p}_{l}(C,\gamma) $ in Eq.~\eqref{eq: low degree approx for noisy prob} close to $\widetilde{p}(C,\gamma)$ in Eq.~\eqref{eq: noisy output probability 2} for each fixed $\gamma \in [0, \gamma_{\rm max}]$ by
\begin{align}\label{eq: lemma: low-degree approximation}
    \Pr_{C \sim \mathcal{H}_{\mathcal{A}}} \left[ \left| \widetilde{p}(C,\gamma) - \widetilde{p}_{l}(C,\gamma)  \right|  > \frac{\varepsilon'}{2^n} \right] \leq \delta    ,
\end{align}
where the approximation error $\varepsilon'$ can be chosen as 
\begin{align}\label{eq: epsilon'}
    \varepsilon'
    := \frac{4N^2\gamma_{\rm max}}{l \delta }\left(\frac{3e\,N\gamma_{\max}}{8l}\right)^{l}.
\end{align}
\end{lemma}
\begin{proof}
    See Appendix~\ref{sec: proof of low degree approximation}
\end{proof}

We now state the main reduction lemma for proving Theorem~\ref{thm 1}, in which we explicitly establish a complexity-theoretic reduction.
Specifically, we show that given access to an oracle for $(\mathcal{A}_0,  \gamma^*)$-\textsc{Noisy-Probability-Estimation} as defined in Definition~\ref{def: noisy probability estimation problem}, one can solve $\mathcal{A}_0$-\textsc{Ideal-Probability-Estimation} in Definition~\ref{def: ideal probability estimation problem}, provided that the oracle error parameters are chosen appropriately.

\begin{lemma}\label{lemma: proof of theorem 1}
Let $\mathcal{A}_0$ be a circuit architecture in Conjecture~\ref{conj: average-case hardness}, and let $d = \mathsf{depth}_{\mathcal{A}_0}(n)$.
%Fix a noise strength $\gamma^* \in [0,1]$.
Let $\mathcal{O}$ be an oracle that solves $(\mathcal{A}_0,  \gamma^*)$-\textsc{Noisy-Probability-Estimation} in Definition~\ref{def: noisy probability estimation problem}.
If the oracle error parameters satisfy
\begin{align}\label{eq: lemma: error bounds 1 main}
\begin{split}
    \varepsilon &=  \poly(n, d, \varepsilon_0^{-1}, \delta_0^{-1})^{-1}e^{-O(nd \gamma^*) }  ,
    \\ 
    \delta &= O\!\left( \frac{\delta_0}{nd\gamma^*  +  \log (nd\varepsilon_0^{-1}\delta_0^{-1})}
    \right)  ,
\end{split}
\end{align} 
then the $\mathcal{A}_0$-\textsc{Ideal-Probability-Estimation} in Definition~\ref{def: ideal probability estimation problem} can be solved within $\rm{BPP}^{\mathcal{O}}$.
\end{lemma}

Using Lemma~\ref{lemma: proof of theorem 1}, the proof of Theorem~\ref{thm 1} is straightforward. 

\begin{proof}[Proof of Theorem~\ref{thm 1}]
Let $\gamma^* = O(\log n / (nd))$ as given in Theorem~\ref{thm 1}.
By Definition~\ref{def: architecture}, the depth satisfies $d=\mathsf{depth}_{\mathcal A_0}(n)\le \poly(n)$.
These imply that the error parameters in Eq.~\eqref{eq: lemma: error bounds 1 main} of Lemma~\ref{lemma: proof of theorem 1} are given by $\varepsilon, \delta = \poly(n,\varepsilon_0^{-1}, \delta_0^{-1})^{-1}$.
Under this condition, there is a polynomial-time reduction from $\mathcal{A}_0$-\textsc{Ideal-Probability-Estimation} to $(\mathcal{A}_0,  \gamma^*)$-\textsc{Noisy-Probability-Estimation}. 
Finally, under the conjecture that $\mathcal{A}_0$-\textsc{Ideal-Probability-Estimation} is \#P-hard (i.e., Conjecture~\ref{conj: average-case hardness}), $(\mathcal{A}_0,  \gamma^*)$-\textsc{Noisy-Probability-Estimation} is also \#P-hard, thereby concluding the proof.    
\end{proof}

%We sketch the proof of Lemma~\ref{lemma: proof of theorem 1} at the end of this section

We conclude this section by sketching the proof of Lemma~\ref{lemma: proof of theorem 1}, which is a key ingredient in the proof of Theorem~\ref{thm 1}.
A detailed proof of Lemma~\ref{lemma: proof of theorem 1} is provided in Appendix~\ref{appendix: section: proof of theorem 1}.

%See Appendix~\ref{appendix: section: proof of theorem 1} for the full proof. 

\begin{proof}[Proof Sketch of Lemma~\ref{lemma: proof of theorem 1}]

%Given an oracle access to $(\mathcal{A}_0,  \gamma^*)$-\textsc{Noisy-Probability-Estimation} that outputs an estimate of $\widetilde{p}(C,\gamma)$, one can also estimate its low-degree approximation $\widetilde{p}_{l}(C,\gamma)$ by Lemma~\ref{lemma: low-degree approximation}.
%Then, by collecting the estimation values of $\widetilde{p}_{l}(C,\gamma)$ for different values of $\gamma \in [\gamma^*, \gamma_{\rm max}]$, one can perform the polynomial interpolation, which in turn enables one to estimate the value $\widetilde{p}_{l}(C,0) = p(C)$.

By Eq.~\eqref{eq: noisy output probability 2}, $\widetilde{p}(C,\gamma)$ is a polynomial in $\gamma$ of degree at most $N=n(d+1)$ and satisfies $\widetilde{p}(C,0)=p(C)$.
Hence, given oracle access to $(\mathcal{A}_0,  \gamma^*)$-\textsc{Noisy-Probability-Estimation}, one can estimate $p(C)$ by performing polynomial interpolation using estimates of $\widetilde{p}(C,\gamma)$ obtained from oracle calls at different noise strengths $\gamma\in[\gamma^* , 1]$.
However, directly interpolating a degree-$N$ representation induces an exponential imprecision blowup $e^{O(nd)}$.
Therefore, we instead use an $l$-degree approximation $\widetilde p_l(C,\gamma)$ constructed in Lemma~\ref{lemma: low-degree approximation}, which satisfies the probabilistic error bound in Eq.~\eqref{eq: lemma: low-degree approximation} and obeys $\widetilde p_l(C,0)=p(C)$.

Now, given oracle access to $(\mathcal A_0,\gamma^*)$-\textsc{Noisy-Probability-Estimation}, we query the oracle at $l+1$ noise values $\gamma_i \in [\gamma^*,\gamma_{\max}]$. 
Then the oracle returns estimates of $\widetilde p(C,\gamma_i)$, and Lemma~\ref{lemma: low-degree approximation} guarantees that these are also good estimates of $\widetilde p_l(C,\gamma_i)$ by Eq.~\eqref{eq: lemma: low-degree approximation}.
Since $\widetilde p_l(C,\gamma)$ has degree at most $l$, by polynomial interpolation, these $l+1$ evaluations determine an interpolating polynomial approximating $\widetilde p_l(C,\gamma)$.
Evaluating the interpolating polynomial at $\gamma=0$ then gives an estimate of $\widetilde p_l(C,0)=p(C)$.

%three sources of error: oracle error in the estimates of $\widetilde p(C,\gamma_i)$, low-degree truncation error in approximating $\widetilde p(C,\gamma_i)$ by $\widetilde p_l(C,\gamma_i)$, and extrapolation error amplified when extending from $[\gamma^*,\gamma_{\max}]$ to $\gamma=0$.
%all of which are controlled by adjusting $\gamma_{\max}$ and $l$. 

The proof reduces to bounding the total induced error and failure probability by appropriately choosing $\gamma_{\max}$ and $l$. 
By setting $\gamma_{\max}= O(\gamma^*)$ and taking
\begin{align}
l
=
O\!\left(
nd\gamma^*
+
\log\!\left(nd\epsilon_0^{-1}\delta_0^{-1}\right)
\right)  ,
\end{align}
we have the small low-degree approximation error $\varepsilon' \ll \varepsilon_0$ in Lemma~\ref{lemma: low-degree approximation}.
Moreover, polynomial interpolation increases the error by a factor of
\begin{align}
e^{O(l)} =
\poly(n, d,\varepsilon_0^{-1},\delta_0^{-1})
e^{O(nd\gamma^*)} .
\end{align}
Finally, the union bound over the $l+1$ interpolation points increases the total failure probability by a factor of order $O(l)$. 
Combining these bounds yields the sufficient conditions on the oracle imprecision and failure probability stated in Eq.~\eqref{eq: lemma: error bounds 1 main}.

\end{proof}

%This yields the same asymptotic scaling of the noise-strength threshold $\gamma^* = O(\log n/(nd))$, but with a worse bound on $\varepsilon$, namely $\varepsilon = \poly(nd, \varepsilon_0^{-1}, \delta_0^{-1})^{-1}\exp(-Cn(d+1)\gamma^*)$ with $C \approx 60$; approximately twice the exponent constant in the present bound $C = 4e^2 \approx 30$ given in Eq.~\eqref{eq: lemma: error bounds 1}.

\section{Proof of Theorem~\ref{thm: complexity decrease}}\label{section: proof of theorem 3}

In this section, we provide an explicit proof of  Theorem~\ref{thm: complexity decrease}, which is crucial to deriving the ultimate classical simulation-hardness statement for noisy RCS in Theorem~\ref{corol: hardness}. 
Specifically, we show that if there exists a polynomial-time classical sampler that approximates noisy RCS at noise strength $\gamma_{1}$ within TVD error $\beta$, then for any $\gamma_2 \geq \gamma_1$ there also exists a polynomial-time classical sampler that approximates noisy RCS at noise strength $\gamma_2$ within the same TVD error $\beta$.
Here, ``simulation" refers to the standard classical approximate-sampling problem~\cite{bouland2019complexity,movassagh2023hardness,bouland2022noise}:
For a fixed architecture $\mathcal{A}$ and noise strength $\gamma$, given an $n$-qubit circuit $C \sim \mathcal{H}_{\mathcal{A}}$ and an error parameter $\beta$, the algorithm samples from a distribution within TVD $\beta$ of the noisy output distribution $\widetilde{p}(C,\gamma,\cdot)$ in Eq.~\eqref{eq: noisy output probability}, in time $\poly(n,1/\beta)$.

Our first observation is that the single-qubit depolarizing noise channel $\mathcal{E}_{\gamma}$ in Definition~\ref{def: noisy quantum circuit} is closed under composition, i.e., the composition of two depolarizing channels is again a depolarizing channel:

\begin{lemma}\label{lem:composition}
For any $\gamma_a,\gamma_b\in[0,1]$, we have
\begin{align}
\label{eq:composition}
\mathcal{E}_{\gamma_a}\circ \mathcal{E}_{\gamma_b}
\;=\;
\mathcal{E}_{\gamma_a+\gamma_b-\gamma_a\gamma_b}.
\end{align}
\end{lemma}

\begin{proof}
Let $\rho$ be any single-qubit state. 
By definition of depolarizing noise channel in Definition~\ref{def: noisy quantum circuit},
\begin{align}
&(\mathcal{E}_{\gamma_a}\circ \mathcal{E}_{\gamma_b})(\rho) \nonumber \\
&=\mathcal{E}_{\gamma_a}\!\left((1-\gamma_b)\rho+\gamma_b\frac{I}{2}\mathrm{Tr}(\rho)\right)\\
&=(1-\gamma_a)\left((1-\gamma_b)\rho+\gamma_b\frac{I}{2}\mathrm{Tr}(\rho)\right)
\;+\;\gamma_a\frac{I}{2}\mathrm{Tr}(\rho)\\
&=(1-\gamma_a)(1-\gamma_b)\rho
\;+\;\Bigl(1-(1-\gamma_a)(1-\gamma_b)\Bigr)\frac{I}{2}\mathrm{Tr}(\rho),
\end{align}
which is exactly $\mathcal{E}_{\gamma'}(\rho)$ with $\gamma'=1-(1-\gamma_a)(1-\gamma_b)=\gamma_a+\gamma_b-\gamma_a\gamma_b$.
\end{proof}

Recall that the single-qubit depolarizing noise channel admits the Pauli-mixture representation
\begin{align}\label{def: noise channel 2}
    \mathcal{E}_{\gamma}(\rho) = 
    \left(1-\frac{3}{4}\gamma\right)\rho + \frac{1}{4}\gamma\sum_{P \in \{X, Y, Z\}} P\rho P . 
\end{align}
Equivalently, $\mathcal{E}_{\gamma}$ can be viewed as the expectation over a random single-qubit Pauli $P \in \{I, X, Y, Z\}$ to $\rho$, where $P$ is drawn from the distribution $q_{\gamma}$ defined by
\begin{align}
\label{eq:qgamma}
q_\gamma(I)=1-\frac{3\gamma}{4},
\qquad
q_\gamma(X)=q_\gamma(Y)=q_\gamma(Z)=\frac{\gamma}{4}.
\end{align}
Let $w_{\gamma}$ denote the product distribution over Pauli paths $s \in \mathsf{P}_n^{d+1}$ in which every single-qubit Pauli in $s$ is drawn independently according to $q_{\gamma}$ in Eq.~\eqref{eq:qgamma}.
Then, for any function $f(s)$ of Pauli path $s$, we have
\begin{align}\label{eq: def: pauli path distribution}
\E_{s \sim w_{\gamma}} [f(s)] =  \sum_{s \in \mathsf{P}_n^{d+1}} \left(1-\frac{3}{4}\gamma \right)^{N - |s|} \left( \frac{1}{4}\gamma \right)^{|s|} f(s) . 
\end{align}

\subsection{Exact sampling at larger noise using a smaller-noise exact sampler}

We now prove the monotonicity reduction for exact sampling: namely, that exact sampling of noisy RCS with smaller noise strength enables exact sampling of noisy RCS for any larger noise strength, in polynomial time.
To this end, we prove the following lemma.

\begin{lemma}
\label{thm:noise-inflation}
Fix $0 \le \gamma_1 \le \gamma_2 \le 1$.
First, let $\gamma_1<1$, and define
\begin{align}
\label{eq:gamma-add}
\gamma_{\mathrm{add}}
\;:=\;
\frac{\gamma_2-\gamma_1}{1-\gamma_1}\in[0,1].
\end{align}
Then, for any circuit $C$ and output $x \in \{0,1\}^{n}$, one has 
\begin{align}
\widetilde{p}(C,\gamma_2, x) = \E_{s \sim w_{\gamma_{\mathrm{add}}}} \left[ \widetilde{p}(C_{s}, \gamma_1, x) \right] , \label{eq:mixture-identity}
\end{align}
where $C_{s}$ denotes the circuit obtained from $C$ by inserting single-qubit Paulis in a Pauli path $s \in \mathsf{P}_{n}^{d+1}$ at the corresponding locations of $C$.
%where $w_{\gamma_{\mathrm{add}}}$ is the product distribution on Pauli paths in which each single-qubit Pauli at each layer is drawn independently according to~\eqref{eq:qgamma} with $\gamma=\gamma_{\mathrm{add}}$.
Moreover, if $\gamma_1 = \gamma_2 =1$, then both $\widetilde{p}(C,\gamma_1, x)$ and $\widetilde{p}(C,\gamma_2, x)$ are the uniform distribution for any $C$ and $x$, and thus the claim is trivial.
\end{lemma}

\begin{proof}
For $\gamma_{\mathrm{add}}$ in Eq.~\eqref{eq:gamma-add}, we have $\gamma_2=\gamma_1+\gamma_{\mathrm{add}}-\gamma_1\gamma_{\mathrm{add}}$ and thus  Lemma~\ref{lem:composition} implies 
\begin{align}
\label{eq:channel-factor}
\mathcal{E}_{\gamma_2}
\;=\;
\mathcal{E}_{\gamma_1}\circ \mathcal{E}_{\gamma_{\mathrm{add}}}   .
\end{align}
Tensoring this identity over all qubits then yields
\begin{align}
\label{eq:tensor-factor}
\mathcal{E}_{\gamma_2}^{\otimes n}
\;=\;
\mathcal{E}_{\gamma_1}^{\otimes n}\circ \mathcal{E}_{\gamma_{\mathrm{add}}}^{\otimes n}   .
\end{align}
Substituting Eq.~\eqref{eq:tensor-factor} into the definition of noisy circuit channel in Eq.~\eqref{eq:noisy-channel} gives 
\begin{align}
\mathcal{N}_{C, \gamma_2}
&=
\bigl(\mathcal{E}_{\gamma_2}^{\otimes n}\bigr) \circ \mathcal{U}^{(d)}
\circ \cdots \circ
\bigl(\mathcal{E}_{\gamma_2}^{\otimes n}\bigr) \circ \mathcal{U}^{(1)} 
\circ \bigl(\mathcal{E}_{\gamma_2}^{\otimes n}\bigr) 
\\
&=
\bigl(\mathcal{E}_{\gamma_1}^{\otimes n}\circ \mathcal{E}_{\gamma_{\mathrm{add}}}^{\otimes n}\bigr)\circ \mathcal{U}^{(d)}
\circ \cdots
\nonumber \\
& \quad \;\;  \cdots  \circ
\bigl(\mathcal{E}_{\gamma_1}^{\otimes n}\circ \mathcal{E}_{\gamma_{\mathrm{add}}}^{\otimes n}\bigr) \circ \mathcal{U}^{(1)}   \circ \bigl(\mathcal{E}_{\gamma_1}^{\otimes n} \circ \mathcal{E}_{\gamma_{\mathrm{add}}}^{\otimes n}\bigr).
\label{eq:expanded}
\end{align}
Now replace each occurrence of $\mathcal{E}_{\gamma_{\mathrm{add}}}$ in Eq.~\eqref{eq:expanded} by its Pauli-mixture representation in Eq.~\eqref{def: noise channel 2}.
By linearity of quantum channels, $\mathcal{N}_{C, \gamma_2} $ can then be written as an expectation (convex sum) over Pauli paths $s \in \mathsf{P}_n^{d+1}$ with each single-qubit Pauli in $s$ drawn independently from $q_{\gamma_{\rm add}}$ in Eq.~\eqref{eq:qgamma}, where conditioned on the specified $s$, it applies the channel $\mathcal{N}_{C_s, \gamma_1}$ with circuit $C_s$ stated in Lemma~\ref{thm:noise-inflation}.
Equivalently, using the Pauli-path distribution $w_{\gamma_{\mathrm{add}}}$ introduced in Eq.~\eqref{eq: def: pauli path distribution}, we obtain
\begin{align}\label{eq:channel-mixture}
\mathcal{N}_{C, \gamma_2}
=
\E_{s \sim w_{\gamma_{\mathrm{add}}}} \left[
\mathcal{N}_{C_{s}, \gamma_1}  \right].
\end{align}
Accordingly, one has 
\begin{align}
     \widetilde{p}(C,\gamma_2, x)
     &=
     \Tr\left[  \ket{x}\!\bra{x} \; \mathcal{N}_{C, \gamma_2} \bigl( \,\ket{0^n}\!\bra{0^n} \,\bigr)  \right] \\
     &= 
     \Tr\left[  \ket{x}\!\bra{x} \; \E_{s \sim w_{\gamma_{\mathrm{add}}}} \left[
     \mathcal{N}_{C_{s}, \gamma_1}\bigl( \,\ket{0^n}\!\bra{0^n} \,\bigr)  \right]   \right]   \\
     &= 
     \E_{s \sim w_{\gamma_{\mathrm{add}}}} \left[ \Tr\left[  \ket{x}\!\bra{x} \;  \mathcal{N}_{C_{s}, \gamma_1} \bigl( \,\ket{0^n}\!\bra{0^n} \,\bigr)  \right]  \right]  \\
     &= 
     \E_{s \sim w_{\gamma_{\mathrm{add}}}} \left[ \widetilde{p}(C_{s}, \gamma_1, x) \right] ,
\end{align}
thus concluding the proof.

\iffalse
Let $\rho_{\mathrm{in}}=\ket{0^n}\!\bra{0^n}$ and denote the output state after noisy quantum circuit as $\rho_{\gamma}(C)= \mathcal{N}_{C,\gamma}(\rho_{\mathrm{in}})$.
Then, applying~\eqref{eq:channel-mixture} to $\rho_{\mathrm{in}}$ similarly gives
\begin{align}\label{eq: output state mixture}
\rho_{\gamma_2}(C)
=
\E_{s \sim w_{\gamma_{\mathrm{add}}}} \left[
\rho_{\gamma_1}(C_s)   \right] .
\end{align}
Finally, measuring the output state in the computational basis $\ket{x}$ is a linear map, hence it preserves convex combinations.
Therefore, Eq.~\eqref{eq:mixture-identity} is straightforward,
\fi

\end{proof}

Consequently, by Lemma~\ref{thm:noise-inflation}, given oracle access to an exact sampler for noisy RCS with noise strength $\gamma_1$, one can exactly sample from the distribution $\widetilde p(C,\gamma_2,\cdot)$ for an arbitrary circuit $C$ and any $\gamma_{2} \in [\gamma_{1}, 1]$.
The procedure is as follows: first draw a Pauli path $s \sim w_{\gamma_{\rm add}}$ with $\gamma_{\rm add}$ as in Eq.~\eqref{eq:gamma-add}; then construct the modified circuit $C_s$ specified by the sampled Pauli path $s$; and finally make a single oracle call to sample $x \sim \widetilde{p}(C_{s}, \gamma_1, x)$.

Here, since $\mathcal{H}_{\mathcal{A}}$ is Pauli-invariant, Pauli insertions specified by $s$ on $C \sim \mathcal{H}_{\mathcal{A}}$ can be absorbed into adjacent Haar-random gates without changing the circuit structure or introducing additional noisy layers. 
Therefore, $C_s$ is also a valid oracle input drawn from the same ensemble.
The overall procedure requires $O(nd)$ operations to sample $s$ and construct $C_s$, together with one oracle query, and hence runs in polynomial time.

\subsection{Extension to approximate sampling}

Indeed, the same reduction also applies to approximate sampling with TVD error, without any blowup in the TVD bound.

Suppose there exists an approximate sampler for noisy RCS with noise strength $\gamma_1$, such that for any input $C$, it samples the output $x \in \{0,1\}^{n}$ according to the distribution $\bar{p}(C,\gamma_1, x)$ that satisfies the TVD bound: 
\begin{align}\label{eq:assumptionforTVD}
    \frac{1}{2}\sum_{x\in\{0,1\}^n}\bigl|\bar{p}(C,\gamma_1, x) - \widetilde{p}(C,\gamma_1, x) \bigr| \leq \beta ,
\end{align}
for a fixed $\beta < 1$.

Now define $\bar{p}(C,\gamma_2, x)$ for $x \in \{0,1\}^{n}$ as the distribution obtained by (i) sampling a Pauli path $s\sim w_{\gamma_{\mathrm{add}}}$ in Eq.~\eqref{eq: def: pauli path distribution} and (ii) sampling $x$ according to $\bar{p}(C_s,\gamma_1, x)$.
By construction,  
\begin{align}\label{eq:p_bar}
    \bar{p}(C,\gamma_2, x) = \E_{s\sim w_{\gamma_{\mathrm{add}}}}\!\left[ \bar{p}(C_s,\gamma_1, x) \right]  .
\end{align}
The TVD between $\bar{p}(C,\gamma_2, x)$ and $\widetilde{p}(C,\gamma_2, x)$ is bounded by 
\begin{align}
    &\frac{1}{2}\sum_{x\in\{0,1\}^n}\bigl|\bar{p}(C,\gamma_2, x) - \widetilde{p}(C,\gamma_2, x) \bigr| \nonumber \\
    &=
    \frac{1}{2}\sum_{x\in\{0,1\}^n}\bigl| \E_{s\sim w_{\gamma_{\mathrm{add}}}}\!\left[ \bar{p}(C_s,\gamma_1, x) - \widetilde{p}(C_s,\gamma_1, x) \right]  \bigr| \label{eq:dcfsd} \\
    &\leq
    \E_{s\sim w_{\gamma_{\mathrm{add}}}}\!\left[ \frac{1}{2}\sum_{x\in\{0,1\}^n}\bigl| \bar{p}(C_s,\gamma_1, x) - \widetilde{p}(C_s,\gamma_1, x) \bigr| \right] \label{eq:convexity of TVD} \\
    &\leq
    \beta   , \label{eq:approximate TVD bound}
\end{align}
where the equality uses Eqs.~\eqref{eq:p_bar} and
\eqref{eq:mixture-identity}, the first inequality follows from convexity of TVD, and the final inequality follows from Eq.~\eqref{eq:assumptionforTVD}.

We can now prove Theorem~\ref{thm: complexity decrease}.

\begin{proof}[Proof of Theorem~\ref{thm: complexity decrease}]
Suppose we are given oracle access to an approximate sampler for noisy RCS at a noise strength $\gamma_1$. 
For an input circuit $C_0 \sim\mathcal{H}_{\mathcal{A}}$ over a fixed architecture $\mathcal{A}$, the sampler outputs $x\in\{0,1\}^n$ according to some distribution $\bar p(C_0,\gamma_1,x)$ satisfying the TVD bound in Eq.~\eqref{eq:assumptionforTVD}.
Using this oracle, given a circuit $C\sim\mathcal{H}_{\mathcal{A}}$ and any $\gamma_{2} \in [\gamma_{1},1]$, consider the following sampling task: 
\begin{enumerate}
    \item[(1)] Sample a Pauli path $s \sim w_{\gamma_{\rm add}}$ with $\gamma_{\rm add}$ defined in Eq.~\eqref{eq:gamma-add}.
    \item[(2)] Construct the corresponding circuit $C_s$ using a sampled $s$.
    \item[(3)] Query the oracle on input $C_s$ and sample $x \sim \bar{p}(C_{s}, \gamma_1, x)$.
\end{enumerate}
By construction, this procedure samples exactly from $\bar{p}(C,\gamma_2, x)$ in Eq.~\eqref{eq:p_bar}, and Eq.~\eqref{eq:approximate TVD bound} implies that $\bar{p}(C,\gamma_2, x)$ approximates $\widetilde{p}(C, \gamma_2, x)$ within the same TVD bound.
Moreover, as argued above, by the Pauli invariance of $C\sim\mathcal{H}_{\mathcal{A}}$, the sampled Pauli operators in $s$ can be absorbed into $C$, so that $C_s \sim \mathcal{H}_{\mathcal{A}}$ and remains a valid oracle input circuit.
Finally, as in the exact case, the overall process requires $O(nd)$ arithmetic operations and a single oracle query, and therefore runs in polynomial time.

\end{proof}

Lastly, many simulability results provide ``relaxed" guarantees over the random choice of circuit (see below), rather than a uniform TVD guarantee for every circuit.
%, although this differs from the sampling problem considered above, which requires a uniform TVD bound for all circuits drawn from $\mathcal{H}_{\mathcal{A}}$, many simulability results assume ``relaxed" TVD guarantees over random circuits as introduced below. 
We therefore also explain how the monotonicity argument extends to samplers satisfying such relaxed approximation guarantees.

Consider first an approximate sampler satisfying an \textit{average} TVD bound, as given in Refs.~\cite{dalzell2024random, deshpande2022tight}.
That is, suppose the sampler outputs samples from $\bar p(C,\gamma_1,x)$ satisfying
\begin{align}
    \E_{C \sim \mathcal{H}_{\mathcal{A}}} \left[ \frac{1}{2}\sum_{x\in\{0,1\}^n}\bigl|\bar{p}(C,\gamma_1, x) - \widetilde{p}(C,\gamma_1, x) \bigr| \right] \leq \beta .
\end{align}
Applying the same sampling procedure as in the proof of Theorem~\ref{thm: complexity decrease} then gives a sampler at any larger noise strength $\gamma_2\geq\gamma_1$ with output distribution $\bar p(C,\gamma_2,x)$ satisfying the same average TVD bound.

Next, consider an approximate sampler satisfying a \textit{probabilistic} TVD bound, as given in Refs.~\cite{napp2022efficient, aharonov2023polynomial}, which samples from $\bar p(C,\gamma_1,x)$ satisfying
\begin{align}\label{eq: probabilistic bound for gamma1}
    \Pr_{C \sim \mathcal{H}_{\mathcal{A}}} \left[ \frac{1}{2}\sum_{x\in\{0,1\}^n}\bigl|\bar{p}(C,\gamma_1, x) - \widetilde{p}(C,\gamma_1, x) \bigr|  \leq \beta \right] 
    \geq 
    1 - \xi.
\end{align}
Let $f(C)$ denote the probability over $s$ that $C_s$ fails the TVD bound in Eq.~\eqref{eq: probabilistic bound for gamma1}. 
Then, since $C$ is Haar-distributed and so is $C_s$, it holds that $\E_C f(C) \leq \xi$. 
By Markov's inequality, for any $\eta>0$, we have $\Pr_{C}[f(C) > \eta] \leq \xi/\eta$.
For circuits with $f(C) \leq \eta$, at most an $\eta$ fraction of $s$ are bad, which contribute at most $\eta$ to the TVD, since TVD is always bounded by $1$. 
Therefore, applying the sampling procedure in Theorem~\ref{thm: complexity decrease} yields, for any $\gamma_2\geq\gamma_1$,
\begin{align}
    \Pr_{C \sim \mathcal{H}_{\mathcal{A}}} \left[ \frac{1}{2}\sum_{x\in\{0,1\}^n}\bigl|\bar{p}(C,\gamma_2, x) - \widetilde{p}(C,\gamma_2, x) \bigr|  \leq \beta + \eta \right] \nonumber \\
    \geq 
    1 - \frac{\xi}{\eta}.
\end{align}
Accordingly, by choosing an appropriate inverse-polynomial $\eta$, the resulting error parameters remain inverse-polynomial whenever $(\beta, \xi)$ are inverse-polynomial.
For example, choosing $\eta=\sqrt{\xi}$ gives TVD error $\beta+\sqrt{\xi}$ and failure probability at most $\sqrt{\xi}$, which remain inverse-polynomial whenever $\beta$ and $\xi$ are inverse-polynomial.

%which satisfies the TVD bound with high probability over random circuits. 

Hence, we have the following corollary for the monotonicity argument. 

\begin{corollary}[Extension to average-case and probabilistic samplers]
The monotonicity result of Theorem~\ref{thm: complexity decrease} extends to average-case approximate samplers satisfying either an average TVD bound or a probabilistic TVD bound with inverse-polynomial failure probability.
\end{corollary}

%in circuit architectures relevant to near-term RCS experiments~\cite{arute2019quantum, wu2021strong, zhu2022quantum,  morvan2024phase, decross2025computational, gao2025establishing}.

\iffalse
Then, by the sampling procedure used in the proof of Theorem~\ref{thm: complexity decrease}, one immediately obtains an approximate sampler for noise strength $\gamma_2\geq \gamma_1$ whose output distribution $\bar p(C,\gamma_2,x)$ satisfies
\begin{align}
    \E_{C \sim \mathcal{H}_{\mathcal{A}}} \left[ \frac{1}{2}\sum_{x\in\{0,1\}^n}\bigl|\bar{p}(C,\gamma_2, x) - \widetilde{p}(C,\gamma_2, x) \bigr| \right] \leq \beta    .
\end{align}
\fi

\section{Conclusion}\label{section: conclusion}

In this work, we characterized the noise-strength threshold below which noisy RCS retains the classical hardness of ideal RCS. 
Specifically, we showed that for circuit architectures for which ideal RCS is conjectured to be classically hard, noisy RCS retains this hardness whenever the noise strength satisfies $\gamma^* = O\left(\log n/(n d)\right)$.
We further proved a monotonicity property of noisy RCS: increasing the noise strength cannot increase the complexity of classical simulation, so efficient simulation at one noise level implies efficient simulation at all larger noise levels.
For architectures lying in the regime covered both by our hardness assumption and by the simulability result of
Ref.~\cite{dalzell2024random}, this bound matches the known simulability scaling, identifying the asymptotic transition scale $\gamma=\Theta(\log n/(nd))$ at the level of its dependence on $n$ and $d$.

We conclude by highlighting important open problems.
A central open problem is to generalize our hardness result to other noise models (e.g., general Pauli noise or non-unital noise) and to broader classes of circuit ensembles. 
Such extensions could lead to stronger hardness results that hold for more realistic implementations of RCS.
The key features of the depolarizing noise used in our analysis are that the noisy output probabilities depend polynomially on the noise parameter $\gamma$ and recover the ideal distribution at $\gamma = 0$.
We therefore expect that our arguments can be extended to more general noise models, provided that the corresponding noisy output probabilities retain a suitable polynomial structure.
Moreover, our analysis uses Haar-random circuits because of their Pauli invariance and the conjectured classical hardness of ideal RCS for this ensemble~\cite{bouland2019complexity, bouland2022noise, kondo2022quantum, movassagh2023hardness, krovi2022average, bouland2025exponential}. 
Hence, our results can be generalized to broader classes of random-circuit ensembles, provided that they exhibit Pauli invariance and admit analogous hardness evidence in the ideal setting.

Our hardness results are asymptotic complexity-theoretic statements.
Accordingly, while our thresholds characterize how the hardness boundary varies as the system size scales, they do not directly provide precise, non-asymptotic noise benchmarks that can certify the quantum advantage in a finite-size experiment.
An important open question, therefore, is how to translate the asymptotic noise thresholds obtained in this work into concrete finite-size benchmarks to assess whether a finite-size experiment has entered a classically intractable regime.

\acknowledgements

B.G. and H.J. were supported by the Korean government (Ministry of Science and ICT~(MSIT)),
the NRF grants funded by the Korea government~(MSIT)~(Nos.~RS-2024-00413957 and RS-2024-00438415),~and the Institute of Information \& Communications Technology Planning \& Evaluation (IITP) grant funded by the Korea government (MSIT) (IITP-2025-RS-2020-II201606 and  IITP-2025-RS-2024-00437191).
C.O. was supported by the NRF Grants (No. RS-2024-00431768 and No. RS-2025-00515456) funded by the Korean government (MSIT) and IITP grants funded by the Korea government (MSIT) (No. RS-2024-00437284, No. IITP-2025-RS-2025-02283189 and No. IITP-2025-RS-2025-02263264) and by Global Partnership Program of Leading Universities in Quantum Science and Technology (RS-2025-08542968) through the National Research Foundation of Korea~(NRF) funded by the Korean government (Ministry of Science and ICT(MSIT)).

%\end{acknowledgements}

\appendix

\section{Low-degree approximation of $\widetilde p(C,\gamma)$}\label{sec: proof of low degree approximation}

In this appendix, we give an explicit proof of Lemma~\ref{lemma: low-degree approximation} in the main text. 
To this end, we first restate the setup for clarity.
Let $N=n(d+1)$ and $\gamma\in[0,\gamma_{\max}]$.
For each $k\in\{0,1,\dots,N\}$, define
\begin{align}\label{eq:wk_Sk_def_A_opt}
w_k(\gamma):=\Bigl(1-\tfrac{3}{4}\gamma\Bigr)^{N-k}\Bigl(\tfrac{\gamma}{4}\Bigr)^k,
\end{align}
and
\begin{align}
S_k(C):=\sum_{\substack{s\in\mathsf{P}_n^{d+1}\\|s|=k}} p(C_s), 
\end{align}
such that the noisy output probability takes the form
\begin{align}\label{eq:noisy_prob_grouped_A_opt}
\widetilde p(C,\gamma)=\sum_{k=0}^{N} w_k(\gamma)\,S_k(C)  ,
\end{align}
where $w_k(0) = \delta_{k0}$.
Throughout the proof, we construct a family of polynomials $\{g_k^{(l)}(\gamma)\}_{k=0}^{N}$ of degree at most $l$ and satisfies
\begin{align}\label{eq:gk_boundary_condition_main_opt}
g_k^{(l)}(0)=\delta_{k0} .
\end{align}
The resulting $l$-degree approximation of $\widetilde p(C,\gamma)$ given by
\begin{align}\label{eq:noisy_prob_l_main_opt}
    \widetilde p_l(C,\gamma):=\sum_{k=0}^{N} g_k^{(l)}(\gamma)S_k(C) , 
\end{align}
satisfies, for each fixed $\gamma\in[0,\gamma_{\max}]$, the probabilistic error bound
\begin{align}\label{eq:lemma1_prob_bound_main_opt}
    \Pr_{C \sim \mathcal{H}_{\mathcal{A}}} \left[ \left| \widetilde{p}(C,\gamma) - \widetilde{p}_{l}(C,\gamma)  \right|  > \frac{\varepsilon'}{2^n} \right] \leq \delta    ,
\end{align}
where the approximation error $\varepsilon'$ can be chosen as
\begin{align}\label{eq:epsprime_main_opt}
    \varepsilon' = \frac{4N^2\gamma_{\rm max}}{l \delta }\left(\frac{3e\,N\gamma_{\max}}{8l}\right)^{l}    .
\end{align}

To construct a family of polynomials $\{g_k^{(l)}(\gamma)\}_{k=0}^{N}$, our core idea is to approximate each coefficient polynomial $w_k(\gamma)$ in Eq.~\eqref{eq:noisy_prob_grouped_A_opt} by truncating its Chebyshev expansion on the interval $\gamma\in[0,\gamma_{\max}]$, as shown in the following lemma.

\begin{lemma}\label{lem:cheb_wk_opt}
Fix $\gamma_{\max}\in(0,1]$ a positive integer $l<N$, and $N\ge3$.
For each $k\in\{0,1,\dots,N\}$, let $w_k(\gamma)$ be defined as in Eq.~\eqref{eq:wk_Sk_def_A_opt}.
Then there exists a polynomial $g_k^{(l)}(\gamma)$ of degree at most $l$ satisfying $g_k^{(l)}(0)=w_k(0)=\delta_{k0}$ such that defining the truncation error by
\begin{align}\label{eq:delta_k_rho_opt}
\Delta_k
:=\sup_{\gamma\in[0,\gamma_{\max}]}\bigl|w_k(\gamma)-g_k^{(l)}(\gamma)\bigr|   ,
\end{align}
and the combined truncation error by
\begin{align}\label{eq:eta_l_def_opt}
\eta_l:=\sum_{k=0}^{N}\binom{N}{k}3^k\,\Delta_k   ,
\end{align}
one has
\begin{align}\label{eq:eta_l_closed_form_opt}
\eta_l
\le
\frac{4N^2\gamma_{\max}}{l}
\left(\frac{3e\,N\gamma_{\max}}{8l}\right)^{l}   .
\end{align}
\end{lemma}

We defer the proof of Lemma~\ref{lem:cheb_wk_opt} to Appendix~\ref{appendix: proof of chebyshev}.
The proof relies on constructing an $l$-degree polynomial by first expanding the high-degree polynomial $w_k(\gamma)$ in the Chebyshev polynomial basis and then truncating the expansion beyond degree $l$, which gives a good approximation for $w_k(\gamma)$, as in standard approximation-theory arguments~\cite{sachdeva2014faster, mason2002chebyshev, borel1928leccons}.

We now prove Lemma~\ref{lemma: low-degree approximation}.

\begin{proof}[Proof of Lemma~\ref{lemma: low-degree approximation}]

%Fix $\gamma\in[0,\gamma_{\max}]$.
From Eqs.~\eqref{eq:noisy_prob_grouped_A_opt} and \eqref{eq:noisy_prob_l_main_opt},
\begin{align}\label{eq:diff_expand_A_opt}
\widetilde p(C,\gamma)-\widetilde p_l(C,\gamma)
=\sum_{k=0}^{N}\bigl(w_k(\gamma)-g_k^{(l)}(\gamma)\bigr)\,S_k(C).
\end{align}
Since $S_k(C)\ge 0$ for all $k$, for any $\gamma\in[0,\gamma_{\max}]$ and $\Delta_k$ defined in Eq.~\eqref{eq:delta_k_rho_opt}, we have
\begin{align}\label{eq:diff_triangle_A_opt}
\bigl|\widetilde p(C,\gamma)-\widetilde p_l(C,\gamma)\bigr|
\le \sum_{k=0}^{N}\Delta_k\,S_k(C) .
\end{align}

Next we bound the circuit average of $S_k(C)$.
For each fixed Pauli path $s$, the circuit $C_s$ is obtained from $C$ by inserting Paulis according to $s$.
Note that by our convention for the circuit architecture $\mathcal{A}$ in Definition~\ref{def: architecture}, for any circuit drawn from $\mathcal{H}_{\mathcal{A}}$, every qubit experiences at least one Haar-random gate throughout the circuit.
Accordingly, by the invariance of the Haar measure under both left- and right-multiplication of Pauli matrices, the ensemble $\mathcal{H}_{\mathcal{A}}$ is Pauli invariant, so that it remains unchanged under Pauli operations on any site of the circuit.

Since $\mathcal{H}_{\mathcal{A}}$ is Pauli-invariant, $C_s$ is distributed identically to $C$. 
Moreover, by invariance of $\mathcal{H}_{\mathcal{A}}$, we have
$\mathbb E_{C\sim\mathcal{H}_{\mathcal{A}}}[p(C)]=2^{-n}$. 
Hence for every fixed $s$,
\begin{align}\label{eq:EpCs_opt}
\mathbb E_{C\sim\mathcal{H}_{\mathcal{A}}}[p(C_s)]=2^{-n}.
\end{align}
Therefore,
\begin{align}\label{eq:ESk_opt}
\mathbb E_{C\sim\mathcal{H}_{\mathcal{A}}}[S_k(C)]
&=\sum_{\substack{s\in\mathsf{P}_n^{d+1}\\|s|=k}}\mathbb E_{C\sim\mathcal{H}_{\mathcal{A}}}[p(C_s)] \\
&=\left|\{s\in\mathsf{P}_n^{d+1}:|s|=k\}\right|\,2^{-n} \\
&=\binom{N}{k}3^k\,2^{-n} ,
\end{align}
and taking expectations of Eq.~\eqref{eq:diff_triangle_A_opt} and using Eq.~\eqref{eq:ESk_opt} yields
\begin{align}
\mathbb E_{C\sim\mathcal{H}_{\mathcal{A}}}\!\left[\bigl|\widetilde p(C,\gamma)-\widetilde p_l(C,\gamma)\bigr|\right]
&\le 2^{-n}\sum_{k=0}^{N}\binom{N}{k}3^k\,\Delta_k \\
&=2^{-n}\,\eta_l, \label{eq:Eabsdiff_opt}
\end{align}
where $\eta_l$ is defined in Eq.~\eqref{eq:eta_l_def_opt}.

Finally, by Markov's inequality, for any $\varepsilon'>0$,
\begin{align}
&\Pr_{C\sim\mathcal{H}_{\mathcal{A}}}\!\left[
\bigl|\widetilde p(C,\gamma)-\widetilde p_l(C,\gamma)\bigr|>\varepsilon'2^{-n}
\right]  \nonumber  \\
&\le 
\frac{2^{n}}{\varepsilon'} \mathbb E_{C\sim\mathcal{H}_{\mathcal{A}}}\!\left[\bigl|\widetilde p(C,\gamma)-\widetilde p_l(C,\gamma)\bigr|\right] \\
&\leq 
\frac{\eta_l}{\varepsilon'}. \label{eq:markov_opt}
\end{align}
Choosing $\varepsilon':=\eta_l/\delta$ gives Eq.~\eqref{eq:lemma1_prob_bound_main_opt}.
Using the explicit bound in Eq.~\eqref{eq:eta_l_closed_form_opt} from Lemma~\ref{lem:cheb_wk_opt} yields Eq.~\eqref{eq:epsprime_main_opt}, completing the proof of Lemma~\ref{lemma: low-degree approximation}.

\end{proof}

\section{Proof of Lemma~\ref{lemma: proof of theorem 1}}\label{appendix: section: proof of theorem 1}

In this appendix, we state a more explicit version of Lemma~\ref{lemma: proof of theorem 1} and provide its detailed proof.

\begin{lemma}\label{lemma: proof of theorem 1 app}
Let $\mathcal{A}_0$ be a circuit architecture in Conjecture~\ref{conj: average-case hardness}, let $d = \mathsf{depth}_{\mathcal{A}_0}(n)$, and let $\gamma^* \in [0,1/3]$.
Let $\mathcal{O}$ be an oracle that solves $(\mathcal{A}_0,  \gamma^*)$-\textsc{Noisy-Probability-Estimation} in Definition~\ref{def: noisy probability estimation problem}.
When the error parameters of the oracle satisfy
\begin{align}\label{eq: lemma: error bounds 1}
\begin{split}
    \varepsilon &= O\!\left(
    \frac{\varepsilon_0^3\delta_0^2}{n^3d^3}\;
    e^{-4e^{2}n(d+1)\gamma^*}
    \right),
    \\ 
    \delta &= O\!\left( \frac{\delta_0}{nd\gamma^*  +  \log (nd\varepsilon_0^{-1}\delta_0^{-1})}
    \right)  ,
\end{split}
\end{align} 
then the $\mathcal{A}_0$-\textsc{Ideal-Probability-Estimation} in Definition~\ref{def: ideal probability estimation problem} can be solved within $\rm{BPP}^{\mathcal{O}}$.
\end{lemma}

\begin{proof}
%[Proof of Lemma~\ref{lemma: proof of theorem 1}]

Throughout the proof, we consider the case $\gamma^* >0$, since the reduction is trivial when $\gamma^* =0$.
To solve $\mathcal{A}_0$-\textsc{Ideal-Probability-Estimation} in Definition~\ref{def: ideal probability estimation problem}, one needs to obtain, for a given $n$-qubit circuit $C \sim \mathcal{H}_{\mathcal{A}_0}$, an estimate $\hat{P}$ of the ideal output probability $p(C)$ within additive error $\varepsilon_0 2^{-n}$ with success probability at least $1 - \delta_0$ over the choice of $C \sim \mathcal{H}_{\mathcal{A}_0}$, such that 
\begin{align}\label{eq: goal}
    \Pr \left[ \left| \hat{P} - p(C) \right|  > \frac{\varepsilon_0}{2^n} \right]  <  \delta_0   .
\end{align}
Let $\mathcal{O}$ be an oracle that solves $(\mathcal{A}_0,  \gamma^*)$-\textsc{Noisy-Probability-Estimation} in Definition~\ref{def: noisy probability estimation problem}, namely, given $C \sim \mathcal{H}_{\mathcal{A}_0}$ and $\gamma \in [\gamma^*,1]$, outputs an estimate of $\widetilde{p}(C,\gamma)$ with high probability such that 
\begin{align}\label{eq: oracle for thm 1}
    \Pr_{C \sim \mathcal{H}_{\mathcal{A}_0}} \left[ \left| \mathcal{O}(C,\gamma) - \widetilde{p}(C,\gamma)  \right|  > \frac{\varepsilon}{2^n} \right] < \delta    .
\end{align}
Throughout the proof, we analyze the conditions on $\varepsilon$ and $\delta$ required to construct an estimate $\hat{P}$ that satisfies Eq.~\eqref{eq: goal} within $\rm{BPP}^{\mathcal{O}}$.

First, by Lemma~\ref{lemma: low-degree approximation}, one can construct the $l$-degree polynomial approximation $\widetilde{p}_{l}(C,\gamma)$ that satisfies 
\begin{align}\label{eq: lemma: low-degree approximation app}
    \Pr_{C \sim \mathcal{H}_{\mathcal{A}}} \left[ \left| \widetilde{p}(C,\gamma) - \widetilde{p}_{l}(C,\gamma)  \right|  > \frac{\varepsilon'}{2^n} \right] \leq \delta    ,
\end{align}
for the approximation error
\begin{align}\label{eq: recap epsilon'}
    \varepsilon'
    := \frac{4N^2\gamma_{\rm max}}{l \delta }\left(\frac{3e\,N\gamma_{\max}}{8l}\right)^{l}.
\end{align} 
%using the polynomials $g_{k}^{(l)}(\gamma)$ in Lemma~\ref{lemma: low-degree approximation}.
Hence, for such $\widetilde{p}_{l}(C,\gamma)$, combining Eq.~\eqref{eq: oracle for thm 1} and Eq.~\eqref{eq: lemma: low-degree approximation app} leads to
\begin{align}
    &\Pr_{C \sim \mathcal{H}_{\mathcal{A}_0}} \left[ \left| \mathcal{O}(C,\gamma) - \widetilde{p}_{l}(C,\gamma)  \right|  > \frac{\varepsilon + \varepsilon' }{2^n} \right]  \nonumber \\
    &\leq \Pr_{C \sim \mathcal{H}_{\mathcal{A}_0}} \left[ \left| \mathcal{O}(C,\gamma) - \widetilde{p}(C,\gamma)  \right|  > \frac{\varepsilon}{2^n} \right]  \nonumber \\
    &+ \Pr_{C \sim \mathcal{H}_{\mathcal{A}_0}} \left[ \left| \widetilde{p}(C,\gamma) - \widetilde{p}_{l}(C,\gamma)  \right|  > \frac{\varepsilon'}{2^n} \right] \\
    &< 2\delta   , \label{eq: estimation of low-degree poly}
\end{align}
where we have used the triangle inequality and the union bound.
%, and $\varepsilon'$ is the approximation error in Lemma~\ref{lemma: low-degree approximation} given by
%\begin{align}\label{eq: recap epsilon'}
%    \varepsilon'
%    = \frac{4N^2\gamma_{\rm max}}{l \delta }\left(\frac{3e\,N\gamma_{\max}}{8l}\right)^{l}.
%\end{align}

To simplify the analysis, we rescale the input variable $\gamma \in [\gamma^*, \gamma_{\rm max}]$ by introducing a new variable $t$ that is linearly related to $\gamma$ as
\begin{align}\label{eq: relation btn gamma and t}
    \gamma(t) = -\frac{\gamma_{\rm max} + \gamma^*}{2}t + \frac{\gamma_{\rm max} + \gamma^*}{2}   , 
\end{align}
so that $\gamma(1) = 0$.
Also, according to Eq.~\eqref{eq: relation btn gamma and t}, the variable $t$ is bounded in the interval $t \in [-\Delta, \Delta]$, where 
\begin{align}
    \Delta = \frac{\gamma_{\rm max} - \gamma^*}{\gamma_{\rm max} + \gamma^*}  ,
\end{align}
which satisfies $\Delta \in (0,1)$ provided that $\gamma_{\rm max} > \gamma^*$.
Note that since $\gamma(t)$ is a linear function of $t$, $\widetilde{p}_{l}(C,\gamma(t))$ is also a polynomial in $t$ of degree at most $l$, that has the ideal value $p(C)$ at $t = 1$.
We also note that  while  $\gamma_{\rm max}$ is not specified at this stage, we assume that $\gamma_{\rm max}$ is strictly larger than $\gamma^*$ and satisfies $\gamma_{\rm max} = \Theta(\gamma^*)$, such that $\Delta$ remains bounded away from both $0$ and $1$.

Let $\{t_{i}\}_{i=0}^{l}$ be a set of equally-spaced points in the interval $[-\Delta, \Delta]$.
For each $t_i$, let $y_{i} = \mathcal{O}(C, \gamma(t_i))$ be an estimation value for the $l$-degree polynomial $\widetilde{p}_{l}(C,\gamma(t_i))$.
Then, by Eq.~\eqref{eq: estimation of low-degree poly}, each $(t_i, y_i)$ satisfies 
\begin{align}\label{eq: ineq for each points}
    \Pr_{C \sim \mathcal{H}_{\mathcal{A}_0}} \left[ \left| y_i - \widetilde{p}_{l}(C,\gamma(t_i))  \right|  > \frac{\varepsilon + \varepsilon' }{2^n} \right]  < 2\delta    .
\end{align}
Using these estimation values, we infer the value at $t = 1$, that is, $\widetilde{p}_{l}(C,\gamma(1)) = p(C)$.
This can be achieved via performing polynomial interpolation, specifically, by constructing the Lagrange interpolating polynomial $P(t)$ from data points $\{(t_i, y_i)\}_{i=0}^{l}$, provided that all $y_i$ points satisfy $ \left| y_i - \widetilde{p}_{l}(C,\gamma(t_i))  \right|  \leq \frac{\varepsilon + \varepsilon' }{2^n}$.
From Eq.~\eqref{eq: ineq for each points} and the union bound, the probability that all $y_i$ points satisfy this condition is at least $1 - 2(l+1)\delta$. 
Given that all the data points are successful, the polynomial interpolation yields an estimate of $\widetilde{p}_{l}(C,\gamma(1)) = p(C)$.
According to~\cite{kondo2022quantum}, the interpolation error for Lagrange polynomial is bounded as follows.

\begin{lemma}[Kondo et al~\cite{kondo2022quantum}]\label{lemma: interpolation}
Let $h(t)$ be a polynomial of degree at most $d$, Let $\Delta\in(0,1)$. Assume that $|h(t_j)|\le \epsilon$ for all of the $d+1$ equally-spaced points $t_j = -\Delta + \frac{2j}{d}\Delta$ for $j = 0,\dots,d$. Then
\begin{align}\label{interpolationerror}
    |h(1)| \leq \epsilon\frac{1}{\sqrt{2\pi d}}\left( \frac{e}{\Delta} \right)^{d} .
\end{align}
\end{lemma}

Consequently, by constructing $l$-degree Lagrange polynomial $P(t)$ using data points $\{(t_i,y_i)\}_{i=0}^{l}$, Lemma~\ref{lemma: interpolation} provides the following error bound:
\begin{align}\label{eq: estimation by interpolation}
    \Pr \left[ \left| P(1) - p(C) \right|  > \frac{\varepsilon + \varepsilon'}{\sqrt{2\pi l}}\left( \frac{e}{\Delta} \right)^{l}\frac{1}{2^n} \right]  < 2(l+1)\delta   .
\end{align}
To satisfy the desired bound in Eq.~\eqref{eq: goal}, which is the ultimate requirement for the reduction, the following conditions on $\delta$ and $\varepsilon$ must hold:
\begin{align}
    2(l+1)\delta \leq \delta_0 \; \rightarrow \; \delta \le \frac{\delta_0}{2(l+1)}  ,      \label{eq: condition for delta 1}
\end{align}
and
\begin{align}
    &\frac{\varepsilon + \varepsilon'}{\sqrt{2\pi l}}\left( \frac{e}{\Delta} \right)^{l}  \leq \varepsilon_0 
    \; \\
    &\rightarrow \;
    \varepsilon \leq 
    \varepsilon_0\sqrt{2\pi l}\left( \frac{e}{\Delta} \right)^{-l}\left(1 - \frac{1}{\sqrt{2\pi l}} \frac{\varepsilon'}{\varepsilon_0 } \left( \frac{e}{\Delta} \right)^{l}\right)  . \label{eq: condition for epsilon 1}
\end{align}
%\varepsilon_0\sqrt{2\pi l}\left( \frac{e}{\Delta} \right)^{-l} - \varepsilon'  
Accordingly, we first set $\delta$ as 
\begin{align}\label{eq: condition for delta 2}
   \delta = \frac{\delta_0}{2(l+1)}  = O(l^{-1}\delta_0).
\end{align}
Next, in the condition on $\varepsilon$ in Eq.~\eqref{eq: condition for epsilon 1}, since $\varepsilon > 0$, it must be promised that $\frac{1}{\sqrt{2\pi l}} \frac{\varepsilon'}{\varepsilon_0 } \left( \frac{e}{\Delta} \right)^{l}$ is strictly smaller than $1$; we find that this holds for a suitable choice of $l$, as stated in the following lemma (all logarithms are natural).

\begin{lemma}\label{lemma: determining l}
    Let $\varepsilon'$ be given in Eq.~\eqref{eq: recap epsilon'} and let $\delta$ be given in Eq.~\eqref{eq: condition for delta 2}.
    Then, the following inequality holds:
    \begin{align}\label{eq: bound for O}
        \frac{1}{\sqrt{2\pi l}} \frac{\varepsilon'}{\varepsilon_0 } \left( \frac{e}{\Delta} \right)^{l} \leq \frac{\sqrt{\gamma_{\max}}}{2}  ,
    \end{align}
    whenever $l$ satisfies 
    \begin{align}\label{eq: def of l}
        l 
        \geq 
        \frac{3e^2 N\gamma_{\max}}{8\Delta}
        + \log\!\left(
        \frac{32}{e}\sqrt{\frac{4\Delta}{3\pi}}\;
        \frac{N^{3/2}}{\varepsilon_0\delta_0}
        \right).
    \end{align}
\end{lemma}

\begin{proof}
    See Appendix~\ref{sec: proof of determining l}
\end{proof}

Hence, when $l$ satisfies Eq.~\eqref{eq: def of l}, the left-hand side of Eq.~\eqref{eq: bound for O} is strictly less than $1$, since $\gamma_{\max} \leq 1$.
To obtain a convenient sufficient bound on $\varepsilon$, we choose $l$ to saturate the bound in Eq.~\eqref{eq: def of l}, with the choice capped at $N$.
If the required $l$ exceeds $N$, we set $l = N$, in which case the approximation is exact and the following bounds are trivial up to constants.
Then, by Lemma~\ref{lemma: determining l}, the condition on $\varepsilon$ in Eq.~\eqref{eq: condition for epsilon 1} reduces to
\begin{widetext}
\begin{align}\label{eq: condition for epsilon 2}
    \varepsilon 
    \leq
    O\!\left( \varepsilon_0\sqrt{N\gamma_{\max} + 
     \log N   +  \log (\varepsilon_0^{-1}\delta_0^{-1})
    }
    \cdot
    e^{-\log\!\left(\frac{e}{\Delta}\right)\left[\frac{3e^2 N\gamma_{\max}}{8\Delta}
    + \log\!\left(\frac{N^{3/2}}{\varepsilon_0\delta_0}
    \right)\right]}
    \right),
\end{align}
where we have used $\log\!\left(\frac{32}{e}\sqrt{\frac{4\Delta}{3\pi}}\right)=O(1)$.
We now specify $\gamma_{\rm max} = 3\gamma^*$, which yields $\Delta = 1/2$ by definition. 
Also, since the square-root term in the right-hand side of Eq.~\eqref{eq: condition for epsilon 2} is larger than $O(1)$, using $N = n(d+1)$, we can further reduce the condition in Eq.~\eqref{eq: condition for epsilon 2} by
\begin{align}\label{eq: condition for epsilon 3}
    \varepsilon 
    \leq 
     O\!\left(  \varepsilon_0(\varepsilon_0\delta_0)^{\log(2e)}
    \left(nd\right)^{-\frac{3}{2}\log(2e)}
    \cdot
    e^{-\frac{9e^2}{4}\log(2e)\,n(d+1)\gamma^*}
    \right)  ,
\end{align}
\end{widetext}
because $\varepsilon$ satisfying Eq.~\eqref{eq: condition for epsilon 3} automatically satisfies Eq.~\eqref{eq: condition for epsilon 2}.

To summarize, given access to the oracle $\mathcal{O}$ for $(\mathcal{A}_0,  \gamma^*)$-\textsc{Noisy-Probability-Estimation}, one can obtain the estimate $P(1)$ of $p(C)$ within ${\rm BPP}^{\mathcal{O}}$ such that
\begin{align}
    \Pr \left[ \left| P(1) - p(C) \right|  > \frac{\varepsilon_0}{2^n} \right]  <  \delta_0   ,
\end{align}
provided that $\varepsilon$ and $\delta$ satisfy the conditions in Eq.~\eqref{eq: condition for epsilon 3} and Eq.~\eqref{eq: condition for delta 2}, respectively.
Moreover, because $\log (2e) < 2$ and $9\log(2e)/4 < 4$, these conditions hold when $\varepsilon$ and $\delta$ are given by 
\begin{align}
\begin{split}
    \varepsilon &= O\!\left(
    \frac{\varepsilon_0^3\delta_0^2}{n^3d^3}\;
    e^{-4e^{2}n(d+1)\gamma^*}
    \right),
    \\ 
    \delta &= O\!\left( \frac{\delta_0}{nd\gamma^*    +  \log (nd\varepsilon_0^{-1}\delta_0^{-1})}
    \right)  ,
\end{split}
\end{align} 
which are the bounds stated in Eq.~\eqref{eq: lemma: error bounds 1}.
Lastly, since $\gamma_{\rm max} = 3\gamma^* \leq 1$, we have the constraint $\gamma^* \leq 1/3$.
This completes the proof.

\end{proof}

\begin{remark}
There exists a simpler proof of Lemma~\ref{lemma: proof of theorem 1 app} based on an alternative $l$-degree polynomial approximation, obtained by first truncating summation terms of $\widetilde{p}(C,\gamma)$ in Eq.~\eqref{eq: noisy output probability 2} over $k \geq l+1$ and then multiplying the resulting expression by $(1 - 3\gamma/4)^{-N+l}$, in analogy with the low-degree truncation method of Ref.~\cite{go2025quantum}. 
This approach yields the bound $\varepsilon = \poly(nd, \varepsilon_0^{-1}, \delta_0^{-1})^{-1}\exp(-Cn(d+1)\gamma^*)$ with $C \approx 60$, which achieves the same asymptotic scaling of the noise-strength threshold $\gamma^* = O(\log n/(nd))$, but with an exponent constant approximately twice that of the present bound $C = 4e^2 \approx 30$ in Eq.~\eqref{eq: lemma: error bounds 1}.
Therefore, although this simpler approach achieves the same asymptotic noise-threshold scaling, we adopt the present approach because it yields a tighter imprecision bound, leaving open the possibility of further optimization.
    
\end{remark}

\section{Proof of Lemma~\ref{lemma: determining l}}\label{sec: proof of determining l}

Recall that the error parameter $\varepsilon'$ in Lemma~\ref{lemma: low-degree approximation} is given by
\begin{align}
    \varepsilon'=\frac{4N^2\gamma_{\max}}{l\delta}\left(\frac{3eN\gamma_{\max}}{8l}\right)^l   .
\end{align}
Based on this definition, one can find that
\begin{align}\label{eq: small}
    \frac{1}{\sqrt{2\pi l}}\frac{\varepsilon'}{\varepsilon_0}\left(\frac{e}{\Delta}\right)^l
    =
    \frac{4N^2\gamma_{\max}}{\varepsilon_0\,l\delta\sqrt{2\pi l}}
    \left(\frac{3e^2N\gamma_{\max}}{8\Delta\,l}\right)^l   .
\end{align}
Now, let us denote
\begin{align}
    A:=\frac{3e^2N\gamma_{\max}}{8\Delta},
    \qquad
    l=A+\chi,
\end{align}
for some yet-to-be-determined parameter $\chi \geq 0$. 
Then we have
\begin{align}
    \left(\frac{A}{l}\right)^l
    &=
    \left(1+\frac{\chi}{A}\right)^{-(A+\chi)} \\
    &=
    \exp\!\left(-(A+\chi)\log\!\left(1+\frac{\chi}{A}\right)\right) \\
    &\le e^{-\chi},  \label{eq: asdfasdf}
\end{align}
where we have used the fact that $(1+x)\log(1+x) \geq x$ for any $x \geq 0$.
Substituting Eq.~\eqref{eq: asdfasdf} into Eq.~\eqref{eq: small} gives
\begin{align}\label{eq: sec bound}
    \frac{1}{\sqrt{2\pi l}}\frac{\varepsilon'}{\varepsilon_0}\left(\frac{e}{\Delta}\right)^l
    \le
    \frac{4N^2\gamma_{\max}}{\varepsilon_0\,l\delta\sqrt{2\pi l}}\;e^{-\chi}   .
\end{align}
Here, recall that by Eq.~\eqref{eq: condition for delta 2}, $\delta$ and $\delta_0$ are related by 
\begin{align}
    \delta = \frac{\delta_0}{2(l+1)} ,
\end{align}
such that 
\begin{align}
    \frac{1}{l\delta} = \frac{2(l+1)}{l \delta_0} \leq \frac{4}{\delta_0},
\end{align}
as long as $l \geq 1$.

Accordingly, starting from Eq.~\eqref{eq: sec bound}, we have 
\begin{align}
    \frac{1}{\sqrt{2\pi l}}\frac{\varepsilon'}{\varepsilon_0}\left(\frac{e}{\Delta}\right)^l
    &\le
    \frac{16N^2\gamma_{\max}}{\varepsilon_0\delta_0\sqrt{2\pi l}}\;e^{-\chi} \\ 
    &\le
    \frac{16N^2\gamma_{\max}}{\varepsilon_0\delta_0\sqrt{2\pi A}}\;e^{-\chi} \\ 
    &=
    \frac{16}{e}\sqrt{\frac{4\Delta}{3\pi}}\;
    \frac{N^{3/2}\sqrt{\gamma_{\max}}}{\varepsilon_0\delta_0}\;e^{-\chi},
\end{align}
where we used $l\ge A$ in the second inequality.
Therefore, for 
\begin{align}
    \chi \ge
    \log\!\left(
    \frac{32}{e}\sqrt{\frac{4\Delta}{3\pi}}\;
    \frac{N^{3/2}}{\varepsilon_0\delta_0}
    \right),
\end{align}
we obtain
\begin{align}
    \frac{1}{\sqrt{2\pi l}}\frac{\varepsilon'}{\varepsilon_0}\left(\frac{e}{\Delta}\right)^l
\le \frac{\sqrt{\gamma_{\rm max}}}{2},
\end{align}
which is exactly Eq.~\eqref{eq: bound for O}. 
Since $l=A+\chi$ gives Eq.~\eqref{eq: def of l}, this completes the proof.

\section{Proof of Lemma~\ref{lem:cheb_wk_opt}}\label{appendix: proof of chebyshev}

We construct $g_k^{(l)}(\gamma)$ by expanding $w_k(\gamma)$ in the Chebyshev basis and truncating the high-order terms.
To simplify the analysis, we map the interval $[0, \gamma_{\rm max}]$ to $t \in [-1,1]$ by 
\begin{align}\label{eq:gamma_to_t_opt}
t=1-\frac{2\gamma}{\gamma_{\max}},
\qquad
\gamma=\frac{\gamma_{\max}}{2}(1-t)    .
\end{align}
With this substitution, we have
\begin{align}\label{eq:wk_affine_t_opt}
w_k(\gamma(t))=\bigl(a_1+b_1 t\bigr)^{N-k}\bigl(a_2+b_2 t\bigr)^{k},
\end{align}
with the coefficients $(a_1,b_1,a_2,b_2)$ given by 
\begin{align}\label{eq:a1b1a2b2_opt}
\begin{split}
&a_1:=1-\frac{3\gamma_{\max}}{8},\quad b_1:=\frac{3\gamma_{\max}}{8},
\\
&a_2:=\frac{\gamma_{\max}}{8},\quad b_2:=-\frac{\gamma_{\max}}{8}  .
\end{split}
\end{align}
Since Eq.~\eqref{eq:wk_affine_t_opt} is a polynomial in $t$ of degree at most $N$, it has an $N$-degree Chebyshev expansion
\begin{align}\label{eq:cheb_expand_opt}
w_k(\gamma(t))=\sum_{j=0}^{N} a_{k,j}\,T_j(t), 
\end{align}
for Chebyshev basis polynomial $T_{j}(t)$.
Now we define $g_k^{(l)}(\gamma)$ by truncating this expansion as
\begin{align}\label{eq:gk_def_opt}
g_k^{(l)}(\gamma)
&:=
\sum_{j=0}^{l} a_{k,j}\,T_j\!\left(1-\frac{2\gamma}{\gamma_{\max}}\right) \nonumber \\
&+\left(\sum_{j=l+1}^{N}a_{k,j}\right)\,
T_l\!\left(1-\frac{2\gamma}{\gamma_{\max}}\right).
\end{align}
This guarantees $\deg(g_k^{(l)})\le l$, and since $T_j(1)=1$ for all $j$,
\begin{align}\label{eq:gk_at0_opt}
g_k^{(l)}(0)=\sum_{j=0}^{N}a_{k,j}=w_k(0)=\delta_{k0} ,
\end{align}
as required.

We next quantify the truncation error $\Delta_{k}$ introduced in Lemma~\ref{lem:cheb_wk_opt}.
For any $t\in[-1,1]$, using $|T_j(t)|\le 1$ and triangle inequality,
\begin{align}
\bigl|w_k(\gamma(t)) - g_k^{(l)}(\gamma(t))\bigr|
&=\left|\sum_{j=l+1}^{N} a_{k,j}\bigl(T_j(t)-T_l(t)\bigr)\right| \\
&\leq \sum_{j = l+1}^{N} \left( \left| a_{k,j} T_j(t) \right| +   \left| a_{k,j} T_l(t) \right| \right) \\
&\le 2\sum_{j=l+1}^{N}|a_{k,j}|. \label{eq:tail_to_coeffs_opt}
\end{align}
A standard tail bound for Chebyshev coefficients~\cite{trefethen2019approximation} implies that for a free parameter $\rho>1$,
\begin{align}\label{eq:cheb_tail_bound_opt}
\sum_{j=l+1}^{\infty}|a_{k,j}|
\le \frac{2M_{k,\rho}}{\rho^{\,l}(\rho-1)},
\end{align}
where
\begin{align}\label{eq:Mkrho_def_opt}
    M_{k,\rho}:=\max_{z\in E_{\rho}}
    \left|\left(a_1+b_1 z\right)^{N-k}\left(a_2+b_2 z\right)^{k}\right|,
\end{align}
and $E_{\rho}$ denotes the Bernstein ellipse in the complex plane with foci located at $\pm 1$, whose semiaxes are $(\rho+\rho^{-1})/2$ and $(\rho-\rho^{-1})/2$. 
Since Eq.~\eqref{eq:tail_to_coeffs_opt} holds for every $t \in [-1,1]$ (thus for every $\gamma \in [0,\gamma_{\rm max}]$), combining Eqs.~\eqref{eq:tail_to_coeffs_opt}--\eqref{eq:cheb_tail_bound_opt} yields
\begin{align}\label{eq:bound_delta_k}
\Delta_k
\le \frac{4}{\rho^{\,l}(\rho-1)}\,M_{k,\rho}   ,
\end{align}
for $\Delta_k$ defined in Eq.~\eqref{eq:delta_k_rho_opt}  .

We now quantify the combined truncation error $\eta_{l}$ introduced in Lemma~\ref{lem:cheb_wk_opt}.
To obtain the bound of $\eta_l$, by Eq.~\eqref{eq:bound_delta_k}, we have
\begin{align}\label{eq:eta_start_opt}
\eta_l
=\sum_{k=0}^{N}\binom{N}{k}3^k\,\Delta_k
\le \frac{4}{\rho^{\,l}(\rho-1)}
\sum_{k=0}^{N}\binom{N}{k}3^k\,M_{k,\rho}   .
\end{align}
To quantify $M_{k,\rho}$, for $z\in E_\rho$, define
\begin{align}\label{eq:A_B_def_opt}
A(z):=a_1+b_1 z,\qquad B(z):=a_2+b_2 z.
\end{align}
By Eq.~\eqref{eq:a1b1a2b2_opt}, $B(z)=\frac{\gamma_{\max}}{8}(1-z)$ and $A(z)=1-\frac{3\gamma_{\max}}{8}(1-z)$, thus we have the constraint
\begin{align}\label{eq:dep_id_opt}
A(z)+3B(z)=1 .
\end{align}
Also, for any fixed $z$,
\begin{align}\label{eq:binom_sum_fixedz_opt}
\sum_{k=0}^{N}\binom{N}{k}3^k\,|A(z)|^{N-k}|B(z)|^{k}
=\bigl(|A(z)|+3|B(z)|\bigr)^{N}.
\end{align}
Moreover, for each $k$,
\begin{align}
\binom{N}{k}3^k\,M_{k,\rho}
&=\max_{z\in E_\rho}\binom{N}{k}3^k\,|A(z)|^{N-k}|B(z)|^k \\
&\le \max_{z\in E_\rho}\bigl(|A(z)|+3|B(z)|\bigr)^{N}. \label{eq:term_vs_sum_opt}
\end{align}
Summing Eq.~\eqref{eq:term_vs_sum_opt} over $k=0,1,\dots,N$ gives
\begin{align}\label{eq:sum_Mk_bound_opt}
\sum_{k=0}^{N}\binom{N}{k}3^k\,M_{k,\rho}
\le (N+1)\,G_\rho^{\,N},
\end{align}
where
\begin{align}\label{eq: def of G_rho}
G_\rho:=\max_{z\in E_\rho}\bigl(|A(z)|+3|B(z)|\bigr).
\end{align}
The maximum in Eq.~\eqref{eq: def of G_rho} is attained at the rightmost point of the Bernstein ellipse $E_\rho$.
To see this clearly, let $c':=\frac{3\gamma_{\max}}8\in(0,\tfrac38]$ and write $u:=z-1$.
Using Eq.~\eqref{eq:dep_id_opt} and $3B(z)=c'(1-z)=-c'u$, we have
\begin{align}\label{eq:G_rho_rewrite_opt}
|A(z)|+3|B(z)|
=|1-3B(z)|+|3B(z)|
=|1+c'u|+c'|u|.
\end{align}
Equivalently, since $u=z-1$,
\begin{align}
|1+c'u|+c'|u|
=
c'\left(
\left|z+\frac{1-c'}{c'}\right|+|z-1|
\right).
\end{align}
Let
\begin{align}
\tau:=-\frac{1-c'}{c'}=1-\frac{1}{c'}.
\end{align}
Since $c'<1/2$, we have $\tau<-1$. 
Therefore, by the triangle inequality,
\begin{align}
|z-\tau|
&\le |z+1|+|-1-\tau| \label{eq: collinear} \\
&= |z+1|+\frac{1}{c'}-2 .
\end{align}
Using the focal property of the Bernstein ellipse,
\begin{align}
|z-1|+|z+1|=\rho+\rho^{-1}
\end{align}
for $z\in E_\rho$, we obtain
\begin{align}
|1+c'u|+c'|u|
&=
c'\left(|z-\tau|+|z-1|\right) \\
&\le
c'\left(|z+1|+|z-1|+\frac{1}{c'}-2\right) \\
&=
1+c'(\rho+\rho^{-1}-2) \\
&=
1+c'\frac{(\rho-1)^2}{\rho}.
\end{align}
Here, equality is attained at the rightmost point
\begin{align}
z_\star=\frac{\rho+\rho^{-1}}{2},
\qquad
u_\star=z_\star-1=\frac{(\rho-1)^2}{2\rho},
\end{align}
because all points $\tau<-1<1<z_\star$ in Eq.~\eqref{eq: collinear} are collinear on the real axis. 
Hence
\begin{align}\label{eq:G_rho_value_opt}
G_\rho
=
|1+c'u_\star|+c'|u_\star|
=
1+2c'u_\star
=
1+\frac{3\gamma_{\max}}{8}\cdot\frac{(\rho-1)^2}{\rho}.
\end{align}
Substituting Eq.~\eqref{eq:sum_Mk_bound_opt} and Eq.~\eqref{eq:G_rho_value_opt} into Eq.~\eqref{eq:eta_start_opt} gives
\begin{align}
\eta_l
&\le
\frac{4(N+1)}{\rho^{\,l}(\rho-1)}
\left(1+\frac{3\gamma_{\max}}{8}\cdot\frac{(\rho-1)^2}{\rho}\right)^{N} \\
&\le
\frac{4(N+1)}{\rho^{\,l}(\rho-1)}
\exp\!\left(\frac{3N\gamma_{\max}}{8}\cdot\frac{(\rho-1)^2}{\rho}\right), \label{eq:eta_rho_preexp_opt}
\end{align}
where the last step uses $\log(1+x)\le x$ for $x\ge 0$.

We now optimize over the free parameter $\rho>1$.
Let
\begin{align}\label{eq:lambda_def_opt}
\lambda:=\frac{3N\gamma_{\max}}{2l}.
\end{align}
Then Eq.~\eqref{eq:eta_rho_preexp_opt} can be written as
\begin{align}\label{eq:eta_lambda_form_opt}
\eta_l
\le
\frac{4(N+1)}{\rho-1}\,
\exp\!\left(-l\left[\log\rho-\lambda\frac{(\rho-1)^2}{4\rho}\right]\right).
\end{align}
Choose $\rho = \rho_\star$ to maximize the bracketed term.
This can be obtained by 
\begin{align}
    \frac{\partial}{\partial\rho}\left( \log\rho - \lambda \frac{(\rho-1)^2}{4\rho} \right)_{\rho = \rho_\star} 
    = 
    \frac{1}{\rho_\star} - \lambda\frac{\rho_\star^2 - 1}{4\rho_\star^2}
    =
    0,   
\end{align}
implying
\begin{align}
    \rho_\star 
    = \frac{2 + \sqrt{\lambda^2 + 4}}{\lambda} .
\end{align}
Note that this definition for $\rho_\star$ gives
\begin{align}
    \log\rho_\star - \lambda \frac{(\rho_\star -1)^2}{4\rho_\star }
    &=
    \log\rho_\star - \frac{\rho_\star - 1}{\rho_\star + 1} \\
    &\geq 
    \log(\frac{\rho_\star}{e})  \\
    &=
    \log( \frac{2 + \sqrt{\lambda^2 + 4}}{e\lambda}  )   . \label{eq: bound for rho_*}
\end{align}
Finally, given $N \geq 3$, choosing $\rho = \rho_\star$ yields
\begin{align}
    \eta_l
    &\le 
    (N+1) \left( \frac{4\lambda}{2+\sqrt{\lambda^2+4} - \lambda}\right) \left(\frac{e\lambda}{2 + \sqrt{\lambda^2 + 4}}\right)^{l} \\
    &\le
    (N+1)\cdot 2\lambda\left(\frac{e\lambda}{4}\right)^{l}\\
    &= 
    (N+1)\cdot \frac{3N\gamma_{\max}}{l}\left(\frac{3e\,N\gamma_{\max}}{8l}\right)^{l} \\
    &\leq 
    \frac{4N^2\gamma_{\max}}{l}\left(\frac{3e\,N\gamma_{\max}}{8l}\right)^{l} , \label{eq:eta_final_opt}
\end{align}
obtaining the desired bound in Eq.~\eqref{eq:eta_l_closed_form_opt}.

Finally, we record an explicit formula for $a_{k,j}$ in Eq.~\eqref{eq:cheb_expand_opt} to explicitly construct $g_k^{(l)}(\gamma)$ in Eq.~\eqref{eq:gk_def_opt}. 
First expand Eq.~\eqref{eq:wk_affine_t_opt} in the power basis:
\begin{align}\label{eq:power_expand_wk_opt}
w_k(\gamma(t))=\sum_{m=0}^{N} c_{k,m}\,t^{m},
\end{align}
where
\begin{align}
c_{k,m}
:=
\sum_{u=\max(0,m-k)}^{\min(m,N-k)} 
&\binom{N-k}{u}\binom{k}{m-u}\,  \nonumber \\
&\times a_1^{\,N-k-u}b_1^{\,u}\,
a_2^{\,k-(m-u)}b_2^{\,m-u}, \label{eq:ckm_opt}
\end{align}
with $(a_1,b_1,a_2,b_2)$ given in Eq.~\eqref{eq:a1b1a2b2_opt}.
Next use the identity for the Chebyshev polynomials~\cite{cody1970survey, mathar2006chebyshev} (for $m\ge 1$)
\begin{align}\label{eq:t_power_to_cheb_opt}
t^{m}&=
\frac{1}{2^{m-1}}\sum_{r=0}^{\lfloor (m-1)/2\rfloor}\binom{m}{r}\,T_{m-2r}(t) \\
& +\frac{1+(-1)^m}{2}\cdot\frac{1}{2^{m}}\binom{m}{m/2}\,T_0(t),
\end{align}
and $t^0=T_0(t)$.
Matching coefficients in Eq.~\eqref{eq:cheb_expand_opt} gives, for $j\ge 1$,
\begin{align}\label{eq:akj_formula_opt}
a_{k,j}=\sum_{r=0}^{\lfloor (N-j)/2\rfloor}
c_{k,j+2r}\,\frac{1}{2^{\,j+2r-1}}\binom{j+2r}{r},
\end{align}
and for $j=0$,
\begin{align}\label{eq:ak0_formula_opt}
a_{k,0}=c_{k,0}+\sum_{r=1}^{\lfloor N/2\rfloor} c_{k,2r}\,\frac{1}{2^{\,2r}}\binom{2r}{r}.
\end{align}
Substituting Eqs.~\eqref{eq:akj_formula_opt}--\eqref{eq:ak0_formula_opt} into Eq.~\eqref{eq:gk_def_opt} completely specifies the polynomials $g_k^{(l)}(\gamma)$.

\bibliographystyle{unsrt}
\bibliography{reference}

\end{document}